%% file: asymmetric-aks.tex
\documentclass[letterpaper,11pt]{article}
\usepackage{times}
\usepackage{fullpage}
\usepackage{hyperref,microtype,pdfsync}
\usepackage{amsmath,amsfonts,amssymb,amsthm,color}
\usepackage{nicefrac}
\usepackage{url}
\usepackage{enumitem}
\usepackage{algorithm}
\usepackage[noend]{algorithmic}

\input{bbhead}

\input{lrhead}


\input{comphead}

\input{thmhead}
\input{noteshead}

\def\eps{\epsilon}

\def\intr{\mathrm{int}}
\def\vx{{\bf x}}
\def\vy{{\bf y}}

\def\vN{{\bf N}}

\DeclareMathOperator*{\vol}{vol}

\newcommand{\pr}[2]{\left\langle #1, #2 \right\rangle }

\title{A $O(1/\eps^2)^n$-time Sieving Algorithm for Approximate Integer Programming}
\author{Daniel Dadush\thanks{H. Milton Stewart School of Industrial and Systems Engineering, Georgia Institute of Technology, 765 Ferst Drive NW, Atlanta, GA 30332-0205, USA {\tt dndadush@gatech.edu}}}

\begin{document}

\maketitle

\begin{abstract}
  \ifnotes \begin{center}{\Huge{NOTES ARE ON}}\end{center} \fi
  \input{abstract}
\end{abstract}





\input{intro}

\input{prelims}

\input{algorithm}

\input{conclusions}

\input{ack}

\bibliographystyle{alphaabbrvprelim}
\bibliography{lattices,acg,cg}

\appendix

\input{appendix}

\end{document}

%% file: bbhead.tex
\newcommand{\Q}{\ensuremath{\mathbb{Q}}}
\newcommand{\R}{\ensuremath{\mathbb{R}}}

\newcommand{\Z}{\ensuremath{\mathbb{Z}}}

%% file: lrhead.tex
\newcommand{\set}[1]{\{{#1}\}}

\newcommand{\floor}[1]{\lfloor{#1}\rfloor}

\newcommand{\ceil}[1]{\lceil{#1}\rceil}


%% file: comphead.tex
\DeclareMathOperator{\poly}{poly}
\DeclareMathOperator{\polylog}{polylog}

\DeclareMathOperator*{\E}{E}



%% file: thmhead.tex
\theoremstyle{plain}            
\newtheorem{theorem}{Theorem}[section]
\newtheorem{lemma}[theorem]{Lemma}
\newtheorem{corollary}[theorem]{Corollary}

\theoremstyle{definition}       
\newtheorem{definition}[theorem]{Definition}

\theoremstyle{remark}           

\numberwithin{equation}{section}

%% file: noteshead.tex
\newif\ifnotes\notesfalse

\ifnotes

\newcommand{\dnote}[1]{{\bf (Daniel:} {#1}{\bf ) }}
\newcommand{\cnote}[1]{{\bf (Chris:} {#1}{\bf ) }}
\newcommand{\snote}[1]{{\bf (Santosh:} {#1}{\bf ) }}

\else

\newcommand{\dnote}[1]{}
\newcommand{\cnote}[1]{}
\newcommand{\snote}[1]{}

\fi

%% file: abstract.tex

The Integer Programming Problem (IP) for a polytope $P \subseteq \R^n$ is to find an integer point in $P$ or decide that
$P$ is integer free. We give a randomized algorithm for an approximate version of this problem, which correctly decides
whether $P$ contains an integer point or whether a $(1+\eps)$ scaling of $P$ around its barycenter is integer free in
time $O(1/\eps^2)^n$ with overwhelming probability. We reduce this approximate IP question to an approximate Closest
Vector Problem (CVP) in a ``near-symmetric'' semi-norm, which we solve via a randomized sieving technique first
developed by Ajtai, Kumar, and Sivakumar (STOC 2001). Our main technical contribution is an extension of the AKS sieving
technique which works for any near-symmetric semi-norm. Our results also extend to general convex bodies and lattices.

%% file: intro.tex
\section{Introduction}
\label{sec:introduction}

The Integer Programming (IP) Problem, i.e. the problem of deciding whether a polytope contains an integer point, is a
classic problem in Operations Research and Computer Science. Algorithms for IP were first developed in the 1950s when
Gomory \cite{gomory-ip58} gave a finite cutting plane algorithm to solve general (Mixed)-Integer Programs. However, the first
algorithms with complexity guarantees (i.e. better than finiteness) came much later. The first such algorithm was the
breakthrough result of Lentra \cite{lenstra83:_integ_progr_with_fixed_number_of_variab}, which gave the first fixed
dimension polynomial time algorithm for IP. Lenstra's approach revolved on finding ``flat'' integer directions of a
polytope, and achieved a leading complexity term of $2^{O(n^3)}$ where $n$ is the number of variables. Lenstra's
approach was generalized and substantially improved upon by Kannan
\cite{kannan87:_minkow_convex_body_theor_and_integ_progr}, who decreased the complexity to $O(n^{2.5})^n$. Recently,
Dadush et al \cite{Dadush-et-al-SVP-11} improved this complexity to $\tilde{O}(n^{\frac{4}{3}})^n$ by using a
solver for the Shortest Vector Problem (SVP) in general norms. Following the works of Lenstra and Kannan, fixed
dimension polynomial algorithms were discovered for counting the number of integer points in a rational polyhedron
\cite{B94}, parametric integer programming \cite{Kannan-Test-Sets-90, Eisenbrand-Shmonin-ParamIP-08}, and integer
optimization over quasi-convex polynomials \cite{Heinz2005, arxiv/HildebrandK10}. However, over the last twenty years
the known algorithmic complexity of IP has only modestly decreased. A central open problem in the area therefore remains
the following:

\paragraph{\emph{\large Question:}} Does there exist a $2^{O(n)}$-time algorithm for Integer Programming? 
\vspace{1em}

In this paper, we show that the answer to this question is affirmative as long as we are willing to accept an
approximate notion of containment. More precisely, we give a randomized algorithm which can correctly distinguish
whether a polytope $P$ contains an integer point or if a small ``blowup'' of $P$ contains no integer points in
$O(1/\eps^2)^n$ time with overwhelming probability. Our results naturally extend to setting of general convex bodies and
lattices, where the IP problem in this context is to decide for a convex body $K$ and lattice $L$ in $\R^n$ whether $K
\cap L = \emptyset$. To obtain the approximate IP result, we reduce the problem to a $(1+\eps)$-approximate Closest
Vector Problem (CVP) under a ``near-symmetric'' semi-norm.

Given a lattice $L \subseteq \R^n$ (integer combinations of a basis $b_1,\dots,b_n \in \R^n$) the SVP is to find
$\min_{y \in L \setminus \set{0}} \|y\|$, and given $x \in \R^n$ the CVP is to find $\min_{y \in L} \|y-x\|$,
where $\|\cdot\|$ is a given (semi-)norm. A semi-norm $\|\cdot\|$ satisfies all norm properties except symmetry, i.e.
we allow $\|x\| \neq \|-x\|$. 

Our methods in this setting are based on a randomized sieving technique first developed by Ajtai, Kumar and Sivakumar
\cite{DBLP:conf/stoc/AjtaiKS01, DBLP:conf/coco/AjtaiKS02} for solving the Shortest (SVP) and Closest Vector Problem
(CVP). In \cite{DBLP:conf/stoc/AjtaiKS01}, they give a $2^{O(n)}$ sieving algorithm for SVP in the $\ell_2$ norm,
extending this in \cite{DBLP:conf/coco/AjtaiKS02} to give a $2^{O(\frac{1}{\eps})}$-time algorithm for $(1+\eps)$-CVP in
the $\ell_2$ norm. In \cite{DBLP:conf/icalp/BlomerN07}, a sieve based $2^{O(n)}$-time algorithm for SVP and
$O(1/\eps^2)^n$-time algorithm for $(1+\eps)$-CVP in any $\ell_p$ norm is given. In \cite{DBLP:conf/fsttcs/ArvindJ08},
the previous results are extended to give a $2^{O(n)}$-time SVP algorithm in any norm (though not semi-norm). In
\cite{conf/socg/EisenbrandHN11}, a technique to boost any $2$-approximation algorithm for $\ell_\infty$ CVP is given
which yields a $O(\ln(\frac{1}{\eps}))^n$ algorithm for $(1+\eps)$-CVP under $\ell_\infty$. Our main technical
contribution is an extension of the AKS sieving technique to give a $O(1/\eps^2)^n$ algorithm for CVP under \emph{any
near-symmetric semi-norm}.

\subsection{Definitions}
\label{sec:definitions}
In what follows, $K \subseteq \R^n$ will denote a convex body (a full dimensional compact convex set) and $L \subseteq
\R^n$ will denote an $n$-dimensional lattice (all integer combinations of a basis of $\R^n$). $K$ will be presented by a
membership oracle in the standard way (see section \ref{sec:prelims}), and $L$ will be presented by a generating basis
$b_1,\dots,b_n \in \R^n$. We define the barycenter (or centroid) of $K$ as $b(K) = \frac{1}{\vol(K)} \int_K x dx$.

For sets $A,B \subseteq \R^n$ and scalars $s,t \in \R$ define the Minkowski Sum $sA + tB = \set{sa + tb: a \in A, b
\in B}$. $\intr(A)$ denotes the interior of the set $A$.

Let $C \subseteq \R^n$ be a convex body where $0 \in \intr(C)$. Define the semi-norm induced
by $C$ (or gauge function of $C$) as $\|x\|_C = \inf \set{s \geq 0: x \in sC}$ for $x \in \R^n$.
$\|\cdot\|_C$ satisfies all norm properties except symmetry, i.e. $\|x\|_C \neq \|-x\|_C$ is allowed. 
$\|\cdot\|_C$ (or $C$) is $\gamma$-symmetric, for $0 < \gamma \leq 1$, if $\vol(C \cap -C) \geq \gamma^n \vol(C)$.
Note $C$ is $1$-symmetric iff $C = -C$. 

For a lattice $L$ and semi-norm $\|\cdot\|_C$, define the first minimum of $L$ under $\|\cdot\|_C$ as $\lambda_1(C,L) =
\inf_{z \in L \setminus \set{0}} \|z\|_C$ (length of shortest non-zero vector). For a target $x$, lattice $L$, and
semi-norm $\|\cdot\|_C$, define the distance from $x$ to $L$ under $\|\cdot\|_C$ as $d_C(L,x) = \inf_{z \in L}
\|z-x\|_C$.

\subsection{Results}
\label{sec:results}

We state our main result in terms of general convex bodies and lattices. We recover the standard integer programming
setting by setting $L = \Z^n$, the standard integer lattice, and $K = \set{x \in \R^n: Ax \leq b}$, a general
polytope. For simplicity, we often omit standard polynomial factors from the runtimes of our algorithms (i.e.
$\polylog$ terms associated with bounds on $K$ or the bit length of the basis for $L$).

Our main result is the following:

\begin{theorem}[Approximate IP Feasibility] 
\label{thm:approx-ip-intro}
For $0 < \eps \leq \frac{1}{2}$, there exists a $O(1/\eps^2)^n$ time algorithm which with probability at least
$1-2^{-n}$ either outputs a point
\[
y \in (1+\eps) K - \eps b(K) ~\cap~ L 
\]
or decides that $K ~\cap~ L = \emptyset$. Furthemore, if 
\[
\frac{1}{1+\eps} K + \frac{\eps}{1+\eps} b(K) \cap L \neq \emptyset \text{,}
\]
the algorithm returns a point $z \in K \cap L$ with probability at least $1-2^{-n}$.
\end{theorem}

The above theorem substantially improves the complexity of IP in the case where $K$ contains a ``deep'' lattice point
(i.e. within a slight scaling of $K$ around its barycenter). Compared to exact algorithms, our methods are competitive
or faster as long as 
\[
\frac{1}{1+n^{-1/2}} K + \frac{n^{-1/2}}{1+n^{-1/2}} b(K) \cap L \neq \emptyset \text{,}
\]
where we achieve complexity $O(n)^n$ (which is the conjectured complexity of the IP algorithm in
\cite{Dadush-et-al-SVP-11}). Hence to improve the complexity of IP below $O(n^{\delta})^n$, for any $0 < \delta <
1$, one may assume that all the integer points lie close to the boundary, i.e. that 
\[
\frac{1}{1+n^{-\frac{1}{2}\delta}}K + \frac{n^{-\frac{1}{2}\delta}}{1+n^{-\frac{1}{2}\delta}}b(K) \cap L = \emptyset \text{.}
\]
The above statement lends credence to the intuition that exact IP is hard because of lattice points lying very near the
boundary.

Starting with the above algorithm, we can use a binary search procedure to go from approximate feasibility to
approximate optimization. This yields the following theorem:

\begin{theorem}[Approximate Integer Optimization]
\label{thm:approx-ip-opt-intro}
For $v \in \R^n$, $0 < \eps \leq \frac{1}{2}$, $\delta > 0$, there exists a $O(1/\eps^2)^n
\polylog(\frac{1}{\delta},\|v\|_2)$ algorithm which with probability at least $1-2^{-n}$ either outputs a point 
\[
y \in K + \eps(K-K) ~\cap~ L
\]
such that
\[
\sup_{z \in K \cap L} \pr{v}{z} \leq \pr{v}{y} + \delta
\]
or correctly decides that $K ~\cap~ L = \emptyset$.
\end{theorem}

The above theorem states that if we wish to optimize over $K ~\cap~ L$, we can find a lattice point in a slight blowup
of $K$ whose objective value is essentially as good as any point in $K ~\cap~ L$. We remark that the blowup is
worse than in Theorem \ref{thm:approx-ip-intro}, since $(1+\eps)K - \eps x \subseteq K + \eps(K-K)$ for any $x \in
K$. This stems from the need to call the feasibility algorithm on multiple restrictions of $K$. To give a clearer
understanding of this notion, the new constraints of the ``blowup'' body can be understood from the following formula:
\[
\sup_{x \in K+\eps(K-K)} \pr{v}{x} = 
           \sup_{x \in K} \pr{v}{x} + \eps\left(\sup_{x \in K} \pr{v}{x} - \inf_{x \in K} \pr{v}{x}\right) \text{.}
\]
Hence each valid constraint $\pr{v}{x} \leq c$ for $K$, is relaxed by an $\eps$-fraction of its $v$'s
variation over $K$. 

\subsection{Main Tool}

We now describe the main tool used to derive both of the above algorithms. At the heart of Theorem
\ref{thm:approx-ip-intro}, is the following algorithm:

\begin{theorem} 
\label{thm:approx-cvp-intro}
Let $\|\cdot\|_C$ denote a $\gamma$-symmetric semi-norm. For $x \in \R^n$, $0 < \eps \leq \frac{1}{2}$, there exists an
$O(\frac{1}{\gamma^4 \eps^2})^n$ time algorithm which computes a point $y \in L$ satisfying
\[
\|y-x\|_C \leq (1+\eps) d_C(L,x)
\]
with probability at least $1-2^{-n}$. Furthermore, if $d_C(L,x) \leq t \lambda_1(C,L)$,
for $t \geq 2$, then an exact closest vector can be found in time $O(\frac{t^2}{\gamma^4})^n$
with probability at least $1-2^{-n}$.
\end{theorem}

The above algorithm adapts the AKS sieve to work for general semi-norms. As mentioned previously
\cite{DBLP:conf/icalp/BlomerN07} gave the above result for $\ell_p$ norms, and \cite{DBLP:conf/fsttcs/ArvindJ08} gave a
$2^{O(n)}$-time exact SVP solver for all norms (also implied by the above since SVP $\leq$ CVP, see
\cite{DBLP:journals/ipl/GoldreichMSS99}). In \cite{Dadush-et-al-SVP-11}, a Las Vegas algorithm (where only the runtime is
probabilistic, not the correctness) is given for the exact versions of the above results (i.e. where an exact closest /
shortest vector is found) with similar asymptotic complexity using completely different techniques,

Hence compared with previous results, the novelty of the above algorithm is the extension of the AKS sieving technique
for $(1+\eps)$-CVP in general semi-norms. As seen from Theorems~\ref{thm:approx-ip-intro} and
~\ref{thm:approx-ip-opt-intro}, the significance of this extension is in its direct applications to IP. Furthermore,
we believe our results illustrate the versatility of the AKS sieving paradigm.

From a high level, our algorithm uses the same framework as \cite{DBLP:conf/icalp/BlomerN07,DBLP:conf/fsttcs/ArvindJ08}.
We first show that the AKS sieve can be used to solve the Subspace Avoiding Problem (SAP), which was first defined in
\cite{DBLP:conf/icalp/BlomerN07}, and use a reduction from CVP to SAP to get the final result. The technical challenge
we overcome, is finding the correct generalizations of the each of the steps performed in previous algorithms to the
asymmetric setting. We discuss this further in section \ref{sec:sap}.

\subsection{Organization}

In section \ref{sec:prelims}, we give some general background in convex geometry and lattices. In section \ref{sec:ip},
we describe the reductions from Approximate Integer Programming to Approximate CVP as well as Approximate Integer
Optimization to Approximate Integer Programming. In section \ref{sec:sap}, we present the algorithm for the Subspace
Avoiding Problem, and in section \ref{sec:cvp} we give the reduction from CVP to SAP. In section \ref{sec:conclusion},
we present our conclusions and open problems.


%% file: prelims.tex
\section{Preliminaries}
\label{sec:prelims}

\paragraph{Computation Model:} 
A convex body $K \subseteq \R^n$ if $(a_0,r,R)$-centered if $a_0 + rB_2^n \subseteq K \subseteq a_0 + RB_2^n$, where
$B_2^n$ is the unit euclidean ball. All the convex bodies in this paper will be $(a_0,r,R)$-centered unless otherwise
specified. To interact with $K$, algorithms are given access to a membership oracle for $K$, i.e. an oracle $O_K$ such
that $O_K(x) = 1$ if $x \in K$ and $0$. In some situations, an exact membership oracle is difficult to implement (e.g.
deciding whether a matrix $A$ has operator norm $\leq 1$), in which situation we settle for a ``weak''-membership
oracle, which only guarantees its answer for points that are either $\eps$-deep inside $K$ or $\eps$-far from $K$ (the
error tolerance $\eps$ is provided as an input to the oracle).

For a $(0,r,R)$-centered $K$ the gauge function $\|\cdot\|_K$ is a semi-norm. To interact with a semi-norm, algorithms
are given a distance oracle, i.e. a function which on input $x$ returns $\|x\|_K$. It is not hard to check that given a
membership oracle for $K$, one can compute $\|x\|_K$ to within any desired accuracy using binary search. Also we remember
that $\|x\|_K \leq 1 \Leftrightarrow x \in K$, hence a distance oracle can easily implement a membership oracle. All the
algorithms in this paper can be made to work with weak-oracles, but for simplicity in presentation, we assume that
our oracles are all exact and that the conversion between different types of oracles occurs automatically. We note that when
$K$ is a polytope, all the necessary oracles can be implemented exactly and without difficulty.

In the oracle model of computation, complexity is measured by the number of oracles calls and arithmetic operations.

\paragraph{Probability:}
For random variables $X,Y \in \Omega$, we define the total variation distance between $X$ and $Y$ as
\[
d_{TV}(X,Y) = \sup_{A \subseteq \Omega} |\Pr(X \in A)-\Pr(Y \in A)|
\]
The following lemma is a standard fact in probability theory:
\begin{lemma} 
\label{lem:tvd}
Let $(X_1,\dots,X_m) \in \Omega^m$ and $(Y_1,\dots,Y_m) \in \Omega^m$ denote independent random variables variables
satisfying $d_{TV}(X_i,Y_i) \leq \eps$ for $i \in [m]$. Then
\[
d_{TV}((X_1,\dots,X_m),(Y_1,\dots,Y_m)) \leq m \eps
\] 
\end{lemma}

\paragraph{Algorithms on Convex Bodies:}

For the purposes of our sieving algorithm, we will need an algorithm to sample uniform points from $K$.
The following result of \cite{DyerFK89} provides the result:

\begin{theorem}[Uniform Sampler]
\label{thm:unif-sampler}
Given $\eta > 0$, there exists an algorithm which outputs a random point $X \in K$ whose distribution has total
variation distance at most $\eta$ from the uniform distribution on $K$, using at most $\poly(n, \ln(\frac{1}{\eta}),
\ln(\frac{R}{r}))$ calls to the oracle and arithmetic operations.
\end{theorem}

We call a random vector $X \in K$ \emph{$\eta$-uniform} if the total variation distance between $X$ and
a uniform vector on $K$ is at most $\eta$.

Our main IP algorithm will provide a guarantee with respect to the barycenter of $K$. The following lemma allows us to
approximate a point near $b(K)$ with overwhelming probability:

\begin{lemma}[Approx. Barycenter]
\label{lem:barycenter-approx}
For $\eps > 0$, let $b = \frac{1}{N} \sum_{i=1}^N X_i$, $N = \frac{c n^2}{\eps^2}$, $c > 0$ an absolute constant, and
where $X_1,\dots,X_N$ are iid $4^{-n}$-uniform samples on $K \subseteq \R^n$. Then
\[
\Pr[\|\pm(b-b(K))\|_{K-b(K)} > \eps] \leq 2^{-n} 
\]
\end{lemma}

\paragraph{Lattices:}
An $n$-dimensional lattice $L \subseteq \R^n$ is formed by integral combinations of linearly independent vectors
$b_1,\dots,b_n \in \R^n$. Letting $B = (b_1,\dots,b_n)$, for a point $x \in \R^n$ we define the modulus operator as
\[
x \bmod B = B(B^{-1}x-\floor{B^{-1} x})
\]
where for $y \in \R^n$, $\floor{y} = (\floor{y_1},\dots,\floor{y_n})$. We note that $x \bmod B \in B[0,1)^n$, i.e. the
fundamental parallelipiped of $B$ and that $x - (x \bmod B) \in L$, hence $x \bmod B$ is the unique representative of
the coset $x + L$ in $B[0,1)^n$.


\paragraph{Convex Geometry:}

The following lemma provides us some simple estimates on the effects of recentering the semi-norm associated
with a convex body.
\begin{lemma}
\label{lem:cent-comp}
Take $x, y \in K$ satisfying $\|\pm(x-y)\|_{K-y} \leq \alpha < 1$. Then for $z \in \R^n$ we have that
\begin{enumerate}
\item $z \in \tau K + (1-\tau)y \Leftrightarrow \|z-y\|_{K-y} \leq \tau$
\item $\|z-y\|_{K-y} \leq \|y-x\|_{K-x} + \alpha ~ \left|1-\|z-x\|_{K-x}\right|$
\item $\|z-x\|_{K-x} \leq \|z-y\|_{K-y} + \frac{\alpha}{1-\alpha} ~ \left|1-\|z-y\|_{K-y}\right|$
\end{enumerate}
\end{lemma}

The following theorem of Milman and Pajor, tells us that $K-b(K)$ is $\frac{1}{2}$-symmetric.
\begin{theorem}[\cite{MP00}]
Assume $b(K) = 0$. Then $\vol(K \cap -K) \geq \frac{1}{2^n} \vol(K)$.
\label{thm:symmetrize}
\end{theorem}

Using the above theorem, we give a simple extension which shows that near-symmetry is a stable property.
\begin{corollary}
Assume $b(K)=0$. Then for $x \in K$ we have that $K-x$ is $\frac{1}{2}(1-\|x\|_{K})$-symmetric.
\label{cor:sym-stable}
\end{corollary}


%% file: algorithm.tex
\section{Algorithms}

\subsection{Integer Programming}
\label{sec:ip}

We describe the basic reduction from Approximate Integer Programming to Approximate CVP, as well as the reduction from
Approximate Integer Optimization to Approximate Integer Programming.

\begin{proof}[Proof of Theorem \ref{thm:approx-ip-intro} (Approximate Integer Programming)]
We are given $0 < \eps \leq \frac{1}{2}$, and we wish to find a lattice point in $(1+\eps)K - \eps b(K) \cap L$ or
decide that $K \cap L = \emptyset$. The algorithm, which we denote by ApproxIP$(K,L,\eps)$, will be the following:

\paragraph{Algorithm:}
\begin{enumerate}
\item Compute $b \in K$, satisfying $\|\pm(b-b(K))\|_{K-b(K)} \leq \frac{1}{3}$, using Lemma~\ref{lem:barycenter-approx}
(see details below).
\item Compute $y \in L$ such that $y$ is $1+\frac{2\eps}{5}$ approximate closest lattice vector to $b$
under the semi-norm $\|\cdot\|_{K-b}$ using Approx-CVP (Theorem~\ref{thm:approx-cvp}).
\item Return $y$ if $y \in \|y-b\|_{K-b} \leq 1+\frac{3\eps}{4}$, and otherwise return ``EMPTY'' (i.e. $K \cap L =
\emptyset$).
\end{enumerate}

\paragraph{Correctness:} Assuming that steps (1) and (2) return correct outputs (which occurs with overwhelming
probability), we show that the final output is correct.

First note that if $\|y-b\|_{K-b} \leq 1+\frac{3\eps}{4}$, then by Lemma~\ref{lem:cent-comp} we have that
\[
\|y-b(K)\|_{K-b(K)} \leq \|y-b\|_{K-b} + \frac{1}{3}\left|1-\|y-b\|_{K-b}\right| \leq 1+\frac{3\eps}{4} +
\frac{1}{3}~\frac{3\eps}{4} = 1+\eps
\]
as required. Now assume that $K \cap L \neq \emptyset$. Then we can take $z \in L$ such that $\|z-b\|_{K-b} \leq 1$.
Since $y$ is a $1+\frac{2\eps}{5}$ closest vector, we must have that $\|y-b\|_{K-b} \leq 1+\frac{2\eps}{5}$. Hence
by the reasoning in the previous paragraph, we have that $\|y-b(K)\|_{K-b(K)} \leq 1+\eps$ as needed.

For the furthermore, we assume that $\frac{1}{1+\eps}K + \frac{\eps}{1+\eps}b(K) \cap L \neq \emptyset$. So we may
pick $z \in L$ such that $\|z-b(K)\|_{K-b(K)} \leq \frac{1}{1+\eps}$. By Lemma~\ref{lem:cent-comp}, we have that
\[
\|z-b\|_{K-b} \leq \|z-b(K)\|_{K-b(K)}+\frac{\frac{1}{3}}{1-\frac{1}{3}}~\left|1-\|z-b(K)\|_{K-b(K)}\right|
              \leq \frac{1}{1+\eps} + \frac{1}{2}~\frac{\eps}{1+\eps} = \frac{1+\frac{\eps}{2}}{1+\eps}
\]
Next by the assumptions on $y$, we have that
$\|y-b\|_{K-b} \leq \frac{1+\frac{\eps}{2}}{1+\eps}(1+\frac{2\eps}{5}) \leq 1$
since $0 < \eps \leq \frac{1}{2}$. Hence $y \in K \cap L$ as needed.

\paragraph{Runtime:}
For step (1), by Lemma~\ref{lem:barycenter-approx} we can compute $b \in K$, satisfying $\|\pm(b-b(K))\|_{K-b(K)} \leq
\frac{1}{3}$, with probability at least $1-2^{-n}$, by letting $b$ be the average of $O(n^2)$ $~4^{-n}$-uniform samples
over $K$. By Theorem \ref{thm:unif-sampler}, each of these samples can be computed in $\poly(n,\ln(\frac{R}{r}))$ time.

For step (2), we first note that by Corollary~\ref{cor:sym-stable}, $K-b$ is $(1-\frac{1}{3})\frac{1}{2} = \frac{1}{3}$-
symmetric. Therefore, the call to the Approximate CVP algorithm, with error parameter $\frac{2\eps}{5}$ returns a valid
approximation vector with probability at least $1-2^{-n}$ in time $O(3 (\frac{5}{2\eps})^2)^n = O(1/\eps^2)^n$.  Hence
the entire algorithm takes time $O(1/\eps^2)^n$ and outputs a correct answer with probability at least $1-2^{-n+1}$ as
needed.
\end{proof}

\begin{proof}[Proof of Theorem \ref{thm:approx-ip-opt-intro} (Approximate Integer Optimization)]
We are given $v \in \R^n$, $0 < \eps \leq \frac{1}{2}$, and $\delta > 0$ where we wish to find a lattice point in $K +
\eps(K-K) ~\cap~ L$ whose objective value is within an additive $\delta$ of the best point in $K ~\cap~ L$. We remember
that $K$ is $(a_0,r,R)$-centered. Since we lose nothing by making $\delta$ smaller, we shall assume that $\delta \leq
\|v\|_2r$.  We will show that Algorithm \ref{alg:approx-opt} correctly solves the optimization problem.

\begin{algorithm}
\caption{Algorithm ApproxOPT$(K, L, v, \eps, \delta)$}
\label{alg:approx-opt}
\begin{algorithmic}[1]

\REQUIRE $(a_0,r,R)$-centered convex body $K \subseteq \R^n$ presented by membership oracle, lattice $L \subseteq \R^n$
given by a basis, objective $v \in \R^n$, tolerance parameters $0 < \eps \leq \frac{1}{2}$ and $\delta > 0$
\ENSURE ``EMPTY'' if $K \cap L = \emptyset$ or $z \in K + \eps(K-K) \cap L$ satisfying $\sup_{y \in K \cap L}
\pr{v}{y} \leq \pr{v}{z} + \delta$

\STATE $z \leftarrow $ApproxIP$(K,L,\eps)$
\IF{z = ``EMPTY''}
   \RETURN  ``EMPTY''
\ENDIF

\STATE Compute $x_l,x_u \in K$ using the ellipsoid algorithm satisfying
$\inf_{x \in K} \pr{v}{x} \geq \pr{v}{x_l} - \frac{\delta}{12}$ and 
$\sup_{x \in K} \pr{v}{x} \leq \pr{v}{x_u} + \frac{\delta}{12}$
\STATE Set $l \leftarrow \pr{v}{z}$ and $u \leftarrow \pr{v}{x_u} + \frac{\delta}{12}$
\WHILE{$u-l > \delta$}
    \STATE $m \leftarrow \frac{1}{2}(u+l)$
	\STATE $y \leftarrow $ApproxIP$(K \cap \set{x \in \R^n: m \leq \pr{v}{x} \leq u},L,\eps)$
	\IF{y = ``EMPTY''}
		\STATE $u \leftarrow m$
		\STATE $y \leftarrow $ApproxIP$(K \cap \set{x \in \R^n: l \leq \pr{v}{x} \leq m},L,\eps)$
		\IF{y = ``EMPTY''}
			\STATE Set $u \leftarrow l$ and $y \leftarrow z$
		\ENDIF
	\ENDIF
	\IF{$\pr{v}{z} < \pr{v}{y}$}
		\STATE Set $z \leftarrow y$ and $l \leftarrow \pr{v}{z}$
	\ENDIF
\ENDWHILE
\RETURN z
\end{algorithmic}
\end{algorithm}

\paragraph{Correctness:}
Assuming that all the calls to the ApproxIP solver output a correct result (which occurs with overwhelming probability),
we show that Algorithm \ref{alg:approx-opt} is correct.  As can be seen, the algorithm performs a standard binary search
over the objective value. During iteration of the while loop, the value $u$ represents the current best upper bound on
$\sup_{y \in K \cap L} \pr{v}{y}$, where this bound is achieved first by bounding $\sup_{x \in K} \pr{v}{x}$ (line 5),
or by showing the lattice infeasibility of appriopriate restrictions of $K$ (line 10 and 13). Similarly, the value $l$
represents the objective value of the best lattice point found thus far, which is denoted by $z$. Now as long as the
value of $z$ is not null, we claim that $z \in K + \eps(K-K)$. To see this note that $z$ is the output of some call to
Approx IP, on $K_{a,b} = K \cap \set{x \in \R^n: a \leq \pr{v}{x} \leq b}$ for some $a < b$, the lattice $L$, with
tolerance parameter $\eps$. Hence if $z$ is none null, we are guaranteed that
\begin{align}
\label{eq:approx-cont}
\begin{split}
z \in (1+\eps)K_{a,b} - \eps b(K_{a,b}) &= K_{a,b}  + \eps(K_{a,b} - b(K_{a,b})) \\
                        &\subseteq K_{a,b} + \eps(K_{a,b} - K_{a,b}) \subseteq K + \eps(K-K)
\end{split}
\end{align}
since $b(K_{a,b}) \subseteq K_{a,b} \subseteq K$. Therefore $z \in K + \eps(K-K)$ as required. Now,
the algorithm returns if ``EMPTY'' if $K \cap L = \emptyset$ (line 3), or $z$ if $u-l < \delta$ (line 17).
Hence the algorithms output is valid as required. 

\paragraph{Runtime:} Assuming that each call to ApproxIP returns a correct result, we first bound the number of
of iterations of the while loop. After this, using a union bound over the failure probability of ApproxIP,
we get a bound on the probability that the algorithm does not perform as described by the analysis.

First, we show that gap $u-l$ decreases by a factor of at least $\frac{3}{4}$ after each iteration of
the loop. Note that by construction of $m$, if $K_{m,u}$ is declared ``EMPTY'' in line 8, then
clearly $u-l$ decreases by $\frac{1}{2}$ in the next step (since $u$ becomes $m$). Next, if a lattice
point $y$ is returned in line 8, we know by Equation \eqref{eq:approx-cont} that $y \in K_{m,u} - \eps(K_{m,u}-K_{m,u})$.
Therefore
\begin{equation}
\label{eq:opt-gap}
\pr{v}{y} \geq \inf_{x \in K_{m,u}} \pr{v}{y} - \eps\left(\sup_{x \in K_{m,u}} \pr{v}{y} - \inf_{x \in K_{m,u}}
\pr{v}{y} \right) \geq m - \eps(u-m) 
\end{equation}
Since $m = \frac{1}{2}(l+u)$, and $\eps \leq \frac{1}{2}$, we see that 
\[
u-\pr{v}{y} \leq (u-m) + \eps(u-m) \leq \frac{1}{2}(u-l) + \frac{1}{4}(u-l) = \frac{3}{4}(u-l)
\]
as needed. From here, we claim that we perform at most $\ceil{\ln(\frac{4R\|v\|_2)}{\ln(\delta)}/\ln(\frac{4}{3})}$
iterations of the for loop.  Now since $K \subseteq a_0 + RB_2^n$, note that the variation of $v$ over $K$ (max minus
min) is at most $2R\|v\|_2$.  Therefore, using Equation~\eqref{eq:opt-gap}, the initial value of $u-l$ (line 6) is at most
\[
2R\|v\|_2 + \eps(2R\|v\|_2) + \frac{\delta}{12} \leq 2R\|v\|_2 + \frac{1}{2}(2R\|v\|_2) + \frac{\|v\|r}{12} \leq 4R\|v\|_2
\]
Since $u-l$ decreases by a factor at least $\frac{3}{4}$ at each iteration, it takes at most
$\ceil{\ln(\frac{4R\|v\|_2}{\delta})/\ln\frac{4}{3}}$ iterations before $u-l \leq \delta$. Since we call ApproxIP at
most twice at each iteration, the probability that any one of these calls fails (wherupon the above analysis does not
hold) is at most $2\ceil{\ln(\frac{4R\|v\|_2}{\delta})/\ln\frac{4}{3}}F$, where $F$ is the failure probability of a call
to ApproxIP. For the purposes of this algorithm, we claim that the error probability of a call to ApproxIP can be made
arbitrarily small by repetition. To see this, note any lattice vector returned by ApproxIP$(K_{a,b},L,\eps)$ is always a
success for our purposes, since by the algorithm's design any returned vector is always in
$K_{a,b}+\eps(K_{a,b}-K_{a,b}) \cap L$ (which is sufficient for us). Hence the only failure possibility is that ApproxIP
returns that $K_{a,b} \cap L = \emptyset$ when this is not the case. By the guarantees on ApproxIP, the probability that
this occurs over $k$ independent repetitions is at most $2^{-nk}$. Hence by repeating each call to ApproxIP
$O(1+\frac{1}{n} \ln\ln\frac{R\|v\|}{\delta})$ times, the total error probability over all calls can be reduced to
$2^{-n}$ as required. Hence with probability at least $1-2^{-n}$, the algorithm correctly terminates in at most
$\ceil{\ln(\frac{4R\|v\|_2}{\delta})/\ln\frac{4}{3}}$ iterations of the while loop.

Lastly, we must check that each call to ApproxIP is done over a well centered $K_{a,b}$, i.e. we must be able to provide
to ApproxIP a center $a_0' \in \R^n$ and radii $r'$, $R'$ such that $a_0' + r'B_2^n \subseteq K_{a,b} \subseteq a_0' +
R'B_2^n$, where each of $a_0',r',R'$ have size polynomial in the input parameters. Here we can show that appropriate
convex combinations of the points $x_l,x_u$ (line 4) and $a_0 \in K$ allow us to get well-centered points inside each
$K_{m,u}$ (line 8) and $K_{l,m}$ (line 11). For simplicity in the presentation, we delay this discussion until the full
version of the paper.

Given the above, since we call ApproxIP at most twice in each iteration (over a well-centered convex body), with
probability at least $1-2^{-n}$ the total running time is $O(\frac{1}{\eps^2})^n \polylog(R,r,\delta,\|v\|_2)$
as required.



\end{proof}

\subsection{Subspace Avoiding Problem}
\label{sec:sap}
In the following two sections, $C \subseteq \R^n$ will denote be a $(0,r,R)$-centered $\gamma$-symmetric convex body,
and $L \subseteq \R^n$ will denote an $n$-dimensional lattice.

In this section, we introduce the Subspace Avoiding Problem of \cite{DBLP:conf/icalp/BlomerN07}, and outline how the AKS
sieve can be adapted to solve it under general semi-norms. We defer most of the analysis to the full version on the paper.

Let $M \subseteq \R^n$ be a linear subspace where $\dim(M) = k \leq n-1$. Let $\lambda(C,L,M) = \inf_{x \in L
\setminus M} \|x\|_C$. Note that under this definition, we have the identity $\lambda_1(C,L) = \lambda(C,L,\set{0})$.

\begin{definition}  The $(1+\eps)$-Approximate Subspace Avoiding Problem with respect $C$, $L$ and $M$ is to find a
lattice vector $y \in L \setminus M$ such that $\|y\|_C = (1+\eps) \lambda(C,L,M)$.
\end{definition}

For $x \in \R^n$, let $\|x\|^*_C = \min \set{\|x\|_C, \|x\|_{-C}}$. For a point $x \in \R^n$, define $s(x) = 1$ if
$\|x\|_C \leq \|x\|_{-C}$ and $s(x) = -1$ if $\|x\|_C > \|x\|_{-C}$. From the notation, we have that $\|x\|_C^* =
 \|x\|_{s(x)C} = \|s(x)x\|_C$.

We begin with an extension of the AKS sieving lemma to the asymmetric setting. The following lemma will provide
the central tool for the SAP algorithm.

\begin{lemma}[Basic Sieve] Let $(x_1,y_1), (x_2,y_2), \dots, (x_N,y_N) \in \R^n \times \R^n$ denote a list
of pairs satisfying $y_i-x_i \in L$, $\|x_i\|_C^* \leq \beta$ and $\|y_i\|_C^* \leq D ~~ \forall i \in [N]$.
Then a clustering, $c: \set{1,\dots,N} \rightarrow J$, $J \subseteq [N]$, satisfying:
\label{lem:basic-sieve}
\[
{1.} ~~ |J| \leq 2\left(\frac{5}{\gamma}\right)^n \quad \quad 
{2.} ~~ \|y_i-y_{c(i)}+x_{c(i)}\|_C^* \leq \frac{1}{2} D + \beta \quad \quad
{3.} ~~ y_i-y_{c(i)}+x_{c(i)}-x_i \in L
\]
for all $i \in [N] \setminus J$, can be computed in deterministic $O(N \left(\frac{5}{\gamma}\right)^n)$-time. 
\end{lemma}
\begin{proof} \hspace{1em}
\paragraph{Algorithm:}
We build the set $J$ and clustering $c$ iteratively, starting from $J = \emptyset$, in the following manner.  For each
$i \in [N]$, check if there exists $j  \in J$ such that $\|y_i-y_j\|_{s(x_j)C} \leq \frac{D}{2}$. If such a $j$ exists,
set $c(i) = j$. Otherwise, append $i$ to the set $J$ and set $c(i) = i$. Repeat.

\paragraph{Analysis:}
We first note, that for any $i,j \in [N]$, we have that $y_i-y_j+x_j-x_i = (y_i-x_i)-(y_j-x_j) \in L$
since by assumption both $y_i-x_i,y_j-x_j \in L$. Hence, property (3) is trivially satisfied by the clustering $c$. 

We now check that the clustering satisfies property (2). For $i \in [N] \setminus J$, note that by construction we have
that $\|y_i-y_{c(i)}\|_{sC} \leq \frac{D}{2}$ where $s = s(x_{c(i)})$. Therefore by the triangle inequality, we have that
\begin{align*}
\|y_i-y_{c(i)}+x_{c(i)}\|_C^* &\leq \|y_i-y_{c(i)}+x_{c(i)}\|_{sC} \leq \|y_i-y_{c(i)}\|_{sC} + \|x_{c(i)}\|_{sC} \\
&= \|y_i-y_{c(i)}\|_{sC} + \|x_{c(i)}\|_C^* \leq \frac{D}{2} + \beta
\end{align*}
as required.

We now show that $J$ satisfies property $(1)$. By construction of $J$, we know that for $i,j \in J$, $i < j$
that $\|y_j-y_i\|_{s(x_i)C} > \frac{D}{2}$. Therefore we have that
\[
\|y_j-y_i\|_{s(x_i)C} > \frac{D}{2} \Rightarrow \|y_j-y_i\|_{C \cap -C} = \|y_i-y_j\|_{C \cap -C} > 
\frac{D}{2} \quad \left(\text{by symmetry of } C \cap -C\right)
\] 
From here, we claim that
\begin{equation}
\label{eq:ls-1}
y_i + \frac{D}{4}(C \cap -C) \cap y_j + \frac{D}{4}(C \cap -C) = \emptyset \text{.}
\end{equation}
Assume not, then we may pick $z$ in the intersection above. Then by definition, we have that
\begin{align*}
\|y_j-y_i\|_{C \cap -C} &= \|(y_j-z)+(z-y_i)\|_{C \cap -C} \leq \|y_j-z\|_{C \cap -C} + \|z-y_i\|_{C \cap -C} \\
                        &= \|z-y_j\|_{C \cap -C} + \|z-y_i\|_{C \cap -C} \leq \frac{D}{4} + \frac{D}{4} = \frac{D}{2}
\end{align*}
a clear contradiction. 

For each $i \in [N]$, we have by assumption that $\|y_i\|_C^* \leq D \Leftrightarrow y_i \in D(C \cup -C)$. Therefore, we
see that
\begin{align}
\label{eq:ls-2}
\begin{split}
y_i + \frac{D}{4}(C \cap -C) &\subseteq D(C \cup -C) + \frac{D}{4}(C \cap -C) \\
                                 &= D((C + \frac{1}{4}(C \cap -C)) \cup (-C + \frac{1}{4}(C \cap -C))) \\
                                 &\subseteq D((C + \frac{1}{4}C) \cup (-C + \frac{1}{4}(-C))) = \frac{5}{4}D(C \cup -C) \\
\end{split}
\end{align}
From \eqref{eq:ls-1}, \eqref{eq:ls-2}, and since $J \subseteq [N]$, we have that 
\begin{align*}
|J| &= \frac{\vol(\set{y_i: i \in J} +\frac{D}{4}(C \cap -C))}{\vol(\frac{D}{4}(C \cap -C))} 
    \leq \frac{\vol(\frac{5}{4}D(C \cup -C))}{\vol(\frac{D}{4}(C \cap -C))} \\
    &\leq \frac{\left(\frac{5}{4}\right)^n(\vol(DC)+\vol(-DC))}{\left(\frac{\gamma}{4}\right)^n \vol(DC)} 
     = 2\left(\frac{5}{\gamma}\right)^n
\end{align*}
as needed.

To bound the running time of the clustering algorithm is straightforward. For each element of $[N]$, we iterate
once through the partially constructed set $J$. Since $|J| \leq 2\left(\frac{5}{\gamma}\right)^n$ throughout the
entire algorithm, we have that the entire runtime is bounded by $O(N \left(\frac{5}{\gamma}\right)^n)$ as required.
\end{proof}

\begin{definition}[Sieving Procedure]
For a list of pairs $(x_1,y_1),\dots,(x_N,y_N)$ as in Lemma \ref{lem:basic-sieve}, we call an application of
the \emph{Sieving Procedure} the process of computing the clustering $c:[N] \rightarrow J$, and outputting the list of pairs
$(x_i, y_i - y_{c(i)} + x_{c(i)})$ for all $i \in [N] \setminus J$. 
\end{definition}
Note that the \emph{Sieving Procedure} deletes the set of pairs associated with the cluster centers $J$, and
combines the remaining pairs with their associated centers.

We remark some differences with the standard AKS sieve. Here the Sieving Procedure does not guarantee that $\|y_i\|_C$
decreases after each iteration. Instead it shows that at least one of $\|y_i\|_C$ or $\|-y_i\|_C$ decreases appropriately
at each step. Hence the region we must control is in fact $D (C \cup -C)$, which we note is generally non-convex.
Additionally, our analysis shows that how well we can use $\|\cdot\|_C$ to sieve only depends on $\vol(C \cap -C) /
\vol(C)$, which is a very flexible global quantity. For example, if $C = [-1,1]^{n-1} \times [-1,2^n]$ (i.e. a cube with
one highly skewed coordinate) then $C$ is still $\frac{1}{2}$-symmetric, and hence the sieve barely notices the asymmetry.

The algorithm for approximate SAP we describe presently will construct a list of large pairs as above, and use repeated
applications of the \emph{Sieving Procedure} to create shorter and shorter vectors.

The next lemma allows us to get a crude estimate on the value of $\lambda(C,L,M)$.
\begin{lemma} 
\label{lem:guess}
Let $C \subseteq \R^n$ a $(0,r,R)$-centered convex body, $L \subseteq \R^n$ be an $n$-dimensional lattice,
and $M \subseteq \R^n$, $\dim(M) \leq n-1$, be a linear subspace. Then a number $\nu > 0$ satisfying
\[
\nu \leq \lambda(C,L,M) \leq 2^n ~ \frac{R}{r} ~ \nu
\]
can be computed in polynomial time.
\end{lemma}

The above lemma follows directly from Lemma 4.1 of \cite{DBLP:journal/tcs/BlomerN09}. They prove it for $\ell_p$
balls, but it is easily adapted to the above setting using the relationship $\frac{1}{r} \|x\|_2 \leq \|x\|_C
\leq \frac{1}{R} \|x\|_2$ (since $C$ is $(0,r,R)$-centered).

The following technical lemma will be needed in the analysis of the SAP algorithm.
\begin{lemma} Take $v \in \R^n$ where $\beta \leq \|v\|_C \leq \frac{3}{2} \beta$. Define $C_v^+ = \beta C \cap (v -
\beta C)$ and $C_v^- = (\beta C - v) \cap - \beta C$. Then 
\[
\frac{\vol(C_v^+)}{\vol(\beta C)} = \frac{\vol(C_v^-)}{\vol(\beta C)} \geq \left(\frac{\gamma}{4}\right)^n 
\]
Furthermore, $\intr(C_v^+) \cap \intr(C_v^-) = \emptyset$.
\label{lem:li}
\end{lemma}

The following is the core subroutine for the SAP solver.
\begin{algorithm}
\caption{ShortVectors(C,L,M,$\beta$,$\eps$)}
\begin{algorithmic}[1]
\REQUIRE $(0,r,R)$-centered $\gamma$-symmetric convex body $C \subseteq \R^n$, basis $B \in \Q^{n \times n}$ for $L$, 
linear subspace $M \subseteq \R^n$, scaling parameter $\beta > 0$, tolerance parameter $0 < \eps \leq \frac{1}{2}$ 
\STATE $D \leftarrow n \max_{1 \leq i \leq n} \|B_i\|_C$
\STATE $\vN_0 \leftarrow 4\ceil{6\ln\left(\frac{D}{\beta}\right)}\left(\frac{20}{\gamma^2}\right)^n 
                         + 8\left(\frac{36}{\gamma^2 \eps}\right)^n$, 
       $~\eta \leftarrow \frac{2^{-(n+1)}}{\vN_0}$
\STATE Create pairs $(\vx_1^0,\vy_1^0)$,$(\vx_2^0,\vy_2^0)$,$~\dots$,$(\vx_{\vN_0}^0,\vy_{\vN_0}^0)$ as follows:
for each $i \in [\vN_0]$,\\ compute $X$ an $\eta$-uniform sample over $\beta C$ (using Theorem \ref{thm:unif-sampler}) and
a uniform $s$ in $\set{-1,1}$, \\
and set $\vx_i^0 \leftarrow s X$ and $\vy_i^0 \leftarrow \vx_i^0 \bmod B$. 
\STATE $t \leftarrow 0$
\WHILE{$D \geq 3\beta$}
   \STATE Apply \emph{Sieving Procedure} to $(\vx_1^t,\vy_1^t)$,$~\dots$,$(\vx_{\vN_t}^t,\vy_{\vN_t}^t)$ yielding
          $(\vx_1^{t+1},\vy_1^{t+1})$,$~\dots$,$(\vx_{\vN_{t+1}}^{t+1},y_{\vN_{t+1}}^{t+1})$
   \STATE $D \leftarrow \frac{D}{2} + \beta$ and $t \leftarrow t + 1$
\ENDWHILE

\RETURN $\set{\vy_i^t-\vx_i^t-(\vy_j^t-\vx_j^t):~ i,j \in [\vN_t]} \setminus M$
\end{algorithmic}
\end{algorithm}

We relate some important details about the the SAP algorithm. Our algorithm for SAP follows a standard procedure. We
first guess a value $\beta$ satisfying $\beta \leq \lambda(C,L,M) \leq \frac{3}{2} \beta$, and then run ShortVectors on
inputs $C,L,M,\beta$ and $\eps$. We show that for this value of $\beta$, ShortVectors outputs a $(1+\eps)$ approximate
solution with overwhelming probability. 

As we can be seen above, the main task of the ShortVectors algorithm, is to generate a large quantity of random vectors,
and sieve them until they are all of relatively small size (i.e. $3\beta \leq 3\lambda(C,L,M)$). ShortVectors then
examines all the differences between the sieved vectors in the hopes of finding one of size $(1+\eps)\lambda(C,L,M)$ in
$L \setminus M$. ShortVectors, in fact, needs to balance certain tradeoffs. On the one hand, it must sieve enough times
to guarantee that the vector differences have small size. On the other, it must use ``large'' perturbations sampled from
$\beta (C \cup -C)$, to guarantee that these differences do not all lie in $M$.

We note that the main algorithmic differences with respect to \cite{DBLP:conf/icalp/BlomerN07, DBLP:conf/fsttcs/ArvindJ08} is
the use of a modified sieving procedure as well as a different sampling distribution for the perturbation vectors (i.e.
over $\beta(C \cup -C)$ instead of just $\beta C$). These differences also make the algorithm's analysis more
technically challenging.

\begin{theorem}[Approximate-SAP] For $0 < \eps \leq \frac{1}{2}$ lattice vector $y \in L \setminus M$ such that
$\|y\|_C \leq (1+\eps) \lambda(C,L,M)$ can be computed in time $O(\frac{1}{\gamma^4 \eps^2})^n$
with probability at least $1-2^{-n}$. Furthermore, if $\lambda(C,L,M) \leq t \lambda_1(C,L)$, $t \geq 2$, a vector $y \in L
\setminus M$ satisfying $\|y\|_C = \lambda(C,L,M)$, can be with computed in time $O\left(\frac{1}{\gamma^4 t^2}\right)^n$
with probability at least $1-2^{-n}$.
\label{thm:approx-sap}
\end{theorem}
\begin{proof}\hspace{1em}
\paragraph{Algorithm:}
The algorithm for $(1+\eps)$-SAP is as follows: 
\begin{enumerate}
\item Using Lemma \ref{lem:guess} compute a value $\nu$ satisfying $\nu \leq \lambda(C,L,M) \leq
2^n \frac{R}{r} \nu$.
\item For each $i \in 0,1,\dots,\ceil{\ln(2^nR/r)/\ln(3/2)}$, let $\beta = (3/2)^i \nu$ and
run ShortVectors$(C,L,\beta,\eps)$.
\item Return the shortest vector found with respect to $\|\cdot\|_C$ in the above runs of ShortVectors.
\end{enumerate}
 
\paragraph{Preliminary Analysis:}
In words, the algorithm first guesses a good approximation of $\lambda(C,L,M)$ (among polynomially many choices)
and runs the ShortVectors algorithm on these guesses. By design, there will be one iteration of the algorithm
where $\beta$ satisfies $\beta \leq \lambda(C,L,M) \leq \frac{3}{2} \beta$. We prove that for this setting
of $\beta$ the algorithm returns a $(1+\eps)$-approximate solution to the SAP problem with probability at least $1-2^{-n}$.

Take $v \in L \setminus M$ denote an optimal solution to the SAP problem, i.e. $v$ satisfies $\|v\|_C = \lambda(C,L,M)$.
We will show that with probability at least $1-2^{-n}$, a small pertubation of $v$ (Claim 4) will be in the set returned by
ShortVectors when run on inputs $C$,$L$,$M$,$\beta$, and $\eps$. 

Within the ShortVectors algorithm, we will assume that the samples generated over $\beta C$ (line 3) are exactly
uniform. By doing this, we claim the probability that ShortVectors returns a $(1+\eps)$-approximate solution to the SAP
problem by changes by at most $2^{-(n+1)}$. To see this, note that we generate exactly $\vN_0$ such samples, all of
which are $\eta$-uniform. Therefore by Lemma \ref{lem:tvd}, we have that the total variation distance between the vector
of approximately uniform samples and truly uniform samples is at most $N_0 \eta = 2^{-(n+1)}$. Lastly, the event that 
ShortVectors returns $(1+\eps)$-approximate solution is a random function of these samples, and hence when
switching uniform samples for $\eta$-uniform ones, the probability of this event changes by at most
$2^{-(n+1)}$. Therefore to prove the theorem, it suffices to show that the failure probability under truly uniform
samples is at most $2^{-(n+1)}$.

In the proof, we adopt all the names of parameters and variables defined in the execution of
ShortVector. We denote the pairs at stage $t$ as $(\vx_1^t,\vy_1^t),$ $\dots$
$,(\vx_{N_t}^t,\vy_{N_t}^t)$. We also let $C_v^+,C_v^-$ be as in Lemma \ref{lem:basic-sieve}. For any stage $t \geq 0$,
we define the pair $(\vx^t_i,\vy^t_i)$, $i \in [\vN_t]$, as \emph{good} if $\vx_i^t \in \intr(C_v^+) \cup \intr(C_v^-)$.

\paragraph{Claim 1:} Let $\mathcal{G}$ denote the event that there are at least
$\frac{1}{2}\left(\frac{\gamma}{4}\right)^n N_0$ \emph{good} pairs at stage $0$. Then $\mathcal{G}$ occurs with
probability least $1-e^{-\frac{1}{48}\gamma^n N_0}$.

Let $G_i = I[\vx_i^0 \in \intr(C_v^+) \cup \intr(C_v^-)]$ for $i \in [\vN_0]$ denote the indicator random variables
denoting whether $(\vx_i^0,\vy_i^0)$ is \emph{good} or not. Let $s_i$, $i \in [N_0]$, denote the $\set{-1,1}$ random
variable indicating whether $\vx_i^0$ is sampled uniformly from $\beta C$ or $-\beta C$. Since $\beta \leq \|v\|_C \leq
\frac{3}{2} \beta$, by lemma \ref{lem:li} we have that
\begin{align*}
\Pr[G_i=1] &\geq \Pr(\vx_i^0 \in \intr(C_v^+) | s_i = 1)\Pr(s_i = 1) + 
                                              \Pr(\vx_i^0 \in \intr(C_v^-) | s_i = -1)\Pr(s_i = -1) \\
     &= \frac{1}{2} \frac{\vol(\beta C \cap (v-\beta C))}{\vol(\beta C)} 
        + \frac{1}{2} \frac{\vol((\beta C-v) \cap (-\beta C))}{\vol(-\beta C)} \geq \left(\frac{\gamma}{4}\right)^n
\end{align*}

From the above we see that $E[\sum_{i=1}^N G_i] \geq \left(\frac{\gamma}{4}\right)^n N_0$. Since the $G_i$'s are
iid Bernoullis, by the Chernoff bound we get that $\Pr[\mathcal{G}] = \Pr[\sum_{i=1}^N G_i < \frac{1}{2}
\left(\frac{\gamma}{4}\right)^n N_0] \leq e^{-\frac{1}{48} \gamma^n N_0}$, as needed.

\paragraph{Claim 2:} Let $T$ denote the last stage of the sieve (i.e. value of $t$ at end of the while loop). Then
conditioned on $\mathcal{G}$, the number of \emph{good} pairs at stage $T$ is at least $N_G =
4\left(\frac{9}{\gamma\eps}\right)^n$.

Examine $(\vx_i^0,\vy^0_i)$ for $i \in [\vN_0]$. We first claim that $\|\vy_i^0\|^*_C \leq D$. To see this note
that $\vy_i^0 = B z$ where $z = B^{-1}\vx_i^0 - \floor{B^{-1}\vx_i^0} \in [0,1)^n$. Hence
\[
\|\vy_i^0\|_C^* \leq \|\vy_i^0\|_C = \|\sum_{i=1}^n B_i z_i\|_C \leq \sum_{i=1}^n z_i \|B_i\|_C \leq n \max_{1 \leq i \leq n} \|B_i\|_C = D
\]
as needed. Let $D_t = \max \set{\|\vy_i^t\|_C^*: i \in [\vN_t]}$, where we note that above shows that $D_0 \leq D$. By Lemma
\ref{lem:basic-sieve}, we know that $\vN_t \geq \vN_{t-1} - 2\left(\frac{5}{\gamma}\right)^n$ and that $D_t \leq
\frac{1}{2}D_{t-1} + \beta$ for $t \geq 1$. For $D_t \geq 3 \beta$, we see that $\frac{1}{2}D_t + \beta \leq
\frac{5}{6} D_t$. Given the previous bounds, an easy computation reveals that $T \leq
\ceil{\frac{\ln(\frac{D}{3\beta})}{\ln(\frac{6}{5})}} \leq \ceil{6\ln(\frac{D}{\beta})}$.

From the above, we see that during the entire sieving phase we remove at most $T(2)\left(\frac{5}{\gamma}\right)^n \leq
2\ceil{6\ln(\frac{D}{\beta})}\left(\frac{5}{\gamma}\right)^n$ pairs. Since we never modify the $\vx_i^t$'s during the sieving
operation, any pair that starts off as \emph{good} stays \emph{good} as long as it survives through the last stage. Since we
start with at least $\frac{1}{2} \left(\frac{\gamma}{4}\right)^n N_0$ \emph{good} pairs in stage $0$, we are 
left with at least 
\begin{align*}
\frac{1}{2} \left(\frac{\gamma}{4}\right)^n N_0 - 2\ceil{6\ln\left(\frac{D}{\beta}\right)}\left(\frac{5}{\gamma}\right)^n
	\geq 4\left(\frac{9}{\gamma \eps}\right)^n = N_G
\end{align*}
\emph{good} pairs at stage $T$ as required.

\paragraph{Modifying the output:}
Here we will analyze a way of modifying the output the ShortVectors, which will maintain the output distribution
but make the output analysis far simpler. 

Let $w(x) = I[x \in \beta C] + I[x \in -\beta C]$. Letting $X$ be uniform in $\beta C$ and $s$ be uniform
in $\set{-1,1}$ (i.e. the distribution in line 3), for $x \in \beta (C \cup -C)$ we have that
\begin{align}
\label{eq:density}
\begin{split}
\mathrm{d}\Pr[sX = x] &= \mathrm{d}\Pr[X = x]\Pr[s = 1] + \mathrm{d}\Pr[X = -x]\Pr[s = -1] = 
\frac{1}{2}\left(\frac{I[x \in C]}{\vol(\beta C)} + \frac{I[x \in -C]}{\vol(\beta C)}\right) \\
&= \frac{w(x)}{2\vol(\beta(C))}
\end{split}
\end{align}

Examine the function $f_v:\beta (C \cup -C) \rightarrow \beta (C \cup -C)$ defined by
\[
f_v(x) = \begin{cases} x-v:& x \in \intr(C_v^+) \\ x+v:& x \in \intr(C_v^-) \\ x:& \text{otherwise} \end{cases}
\]
Since $\intr(C_v^+) \cap \intr(C_v^-) = \emptyset$, it is easy to see that $f_v$ is a well-defined bijection on $\beta
(C \cup - C)$ satisfying $f_v(f_v(x)) = x$. Furthemore by construction, we see that $f_v(\intr(C_v^+)) = \intr(C_v^-)$
and $f_v(\intr(C_v^-)) = \intr(C_v^+)$. Lastly, note that for any $x \in \beta (C \cup -C)$, that $f_v(x) \equiv x \bmod
B$ since $f_v(x)$ is just a lattice vector shift of $x$.

Let $F_v$ denote the random function where 
\[
F_v(x) = \begin{cases} x &: \text{with probability } ~~ \frac{w(x)}{w(x) + w(f_v(x))} \\
                       f_v(x) &: \text{with probability } ~~ \frac{w(f_v(x))}{w(x) + w(f_v(x))} \end{cases}
\]
Here, we intend that different applications of the function $F_v$ all occur with independent randomness.
Next we define the function $c_v$ as
\[
c_v(x,y) = \begin{cases} f_v(x) &: \|y-f_v(x)\|_C^* < \|y-x\|_C^* \\ x &: \text{otherwise} \end{cases}
\]

For any stage $t \geq 0$, define $\bar{\vx}^t_i = c_v(\vx^t_i, \vy^t_i)$.

\paragraph{Claim 3:} For any stage $t \geq 0$, the pairs $(\vx_1^t,\vy_1^t),\dots,(\vx_{N_t}^t,\vy_{N_t}^t)$, and
$(F_v(\bar{\vx}_1^t),\vy_1^t),\dots,(F_v(\bar{\vx}_{N_t}^t), \vy_{N_t}^t)$ are identically distributed. Furthermore, this remains true after conditioning on the event $\mathcal{G}$.

To prove the claim, it suffices to show that the pairs in (1) $(\vx_1^t,\vy_1^t),\dots,(\vx_{N_t}^t,\vy_{N_t}^t)$, and
\linebreak (2) $(F_v(\bar{\vx}_1^t),\vy_1^t),(\vx_2^t,\vy_2^t),\dots,(\vx_{N_t}^t,\vy_{N_t}^t)$ are identically
distributed (both before and after conditioning on $\mathcal{G}$). Our analysis will be independent of the index
analyzed, hence the claim will follow by applying the proof inductively on each remaining pair in the second list.

The pairs in $(1)$ correspond exactly to the induced distribution of the algorithm on the stage $t$ variables. We think
of the pairs in $(2)$ as the induced distribution of a modified algorithm on these variables, where the modified
algorithm just runs the normal algorithm and replaces $(x_1^t,y_1^t)$ by $(F_v(c_v(x_1^t,y_1^t)),y_1^t)$ in stage $t$. To
show the distributional equivalence, we show a probability preserving correspondance between runs of the normal and
modified algorithm having the same stage $t$ variables.

For $0 \leq k \leq t$, let the pairs $(x_i^k,y_i^k)$, $i \in [N_k]$, denote a valid run of the normal algorithm through
stage $t$. We label this as run $A$. Let us denote the sequence of ancestors of $(x_1^t,y_1^t)$ in the normal algorithm by
$(x_{a_k}^k,y_{a_k}^k)$ for $0 \leq k \leq t-1$. By definition of this sequence, we have that $x_{a_0}^0 = x_{a_1}^1 =
\dots = x_1^t$. Since the ShortVectors algorithm is deterministic given the initial samples, the probability density
of this run is simply
\begin{equation}
\label{eq:run-a-prob}
\mathrm{d}\Pr\left[\cap_{i \in [N_0]} \set{\vx_i^0 = x_i^0}\right] = \mathrm{d}\Pr\left[\vx_{a_0}^0 = x_1^t\right]
\prod_{i \in [N_0], i \neq a_0} \mathrm{d}\Pr\left[\vx_i^0 = x_i^0\right]
\end{equation}
by the independence of the samples and since $x_{a_0}^0 = x_1^t$. Notice if we condition on the event $\mathcal{G}$,
assuming the run $A$ belongs to $\mathcal{G}$ (i.e. that there are enough \emph{good} pairs at stage $0$), the above
probability density is simply divided by $\Pr[\mathcal{G}]$.

If $x_1^t \notin \intr(C_v^+) \cup \intr(C_v^-)$, note that $F_v(c_v(x_1^t,y_1^t)) = F_v(x_1^t) = x_1^t$, i.e. the
action of $F_v$ and $c_v$ are trivial. In this case, we associate run $A$ with the identical run for the modified
algorithm, which is clearly valid and has the same probability. Now assume that $x_1^t \in C_v^+ \cup C_v^-$. In this
case, we associate run $A$ to two runs of the modified algorithm: $\tilde{A}$. identical to run $A$, $C$. run $A$ with
$(x_{a_k}^k,y_{a_k}^k)$ replaced by $(f_v(x_{a_k}^k),y_{a_k}^k)$ for $0 \leq k \leq t-1$. Note that both of the
associated runs have the same stage $t$ variables as run $A$ by construction. 

We must check that both runs are indeed valid for the modified algorithm. To see this, note that up till stage $t$, the
modified algorithm just runs the normal algorithm. Run $\tilde{A}$ inherents validity for these stages from the fact
that run A is valid for the normal algorithm. To see that run $C$ is valid, we first note that $f_v(x_{a_0}^0) \equiv
x_{a_0}^0 \equiv y_{a_0}^0 \bmod B$ ($B$ is the lattice basis), which gives validity for stage $0$. By design of the
normal sieving algorithm, note that during run $A$, the algorithm never inspects the contents of $x_{a_k}^k$ for $0 \leq
k < t$. Therefore, if $(x_{a_k}^k,y_{a_k}^k)$ denotes a valid ancestor sequence in run $A$, then so does
$(f_v(x_{a_k}),y_{a_k}^k)$ in run $B$ for $0 \leq k < t$. For stage $t$, note that the normal algorithm, given the stage
$0$ inputs of run $\tilde{A}$ would output $(x_1^t,y_1^t),$ $\dots$$,(x_{N_t},y_{N_t})$ for the stage $t$ variables, and
that given the stage $0$ inputs of run $C$ would output $(f_v(x_1^t),y_1^t),(x_2^t,y_2^t),$ $\dots$ 
$,(x_{N_t}^t,y_{N_t}^t)$. Hence in run $\tilde{A}$, the modified algorithm retains the normal algorithms output, and in
run $C$, it swiches it from $f_v(x_1^t)$ back to $x_1^t$. Therefore, both runs are indeed valid for the modified
algorithm. Furthermore, note that if run $A$ is in $\mathcal{G}$, then both the stage $0$ variables of $\tilde{A}$ and
$C$ index a \emph{good} run for the normal algorithm since the pair $(x_{a_0}^0,y_{a_0})$ is \emph{good} iff
$(f_v(x_{a_0}^0),y_{a_0})$ is \emph{good}. Hence we see that correspondance described is valid both before and after
conditioning on $\mathcal{G}$. Lastly, it is clear that for any run of the modified algorithm, the above correspondance
yields a unique run of the normal algorithm.

It remains to show that the correspondance is probability preserving. We must therefore compute the probability density
associated with the union of run $\tilde{A}$ and $C$ for the modified algorithm. Using the analysis from the previous
paragraph and the computation in~\eqref{eq:run-a-prob}, we see that this probability density is
\begin{align}
\label{eq:corr-prob}
\begin{split}
\Big(\mathrm{d}\Pr[\vx_{a_0}^0 = x_1^t]\Pr[F_v(c_v(x_1^t,y_1^t)) = &x_1^t] 
 + \mathrm{d}\Pr[\vx_{a_0}^0 = f_v(x_1^t)]\Pr[F_v(c_v(f_v(x_1^t),y_1^t)) = x_1^t]\Big) \\
&\prod_{i \in [N_0], i \neq a_0} d\Pr\left[\vx_i^0 = x_i^0\right]
\end{split}
\end{align}
On the first line above, the first term corresponds to run $\tilde{A}$ which samples $x_1^t$ in stage $0$ and then
chooses to keep $(x_1^t,y_1^t)$ in stage $t$, and the second term corresponds to $C$ which samples $f_v(x_1^t)$ in
stage $0$ and chooses to flip $(f_v(x_1^t),y_1^t)$ to $(x_1^t,y_1^t)$ in stage $t$. Now by definition of $F_v$ and $c_v$,
we have that
\[
\Pr[F_v(c_v(f_v(x_1^t),y_1^t)) = x_1^t] = \Pr[F_v(c_v(x_1^t,y_1^t)) = x_1^t] = \frac{w(x_1^t)}{w(x_1^t)+w(f_v(x_1^t))}
\]
Therefore, using the above and Equation~\eqref{eq:density}, we have that the first line of~\eqref{eq:corr-prob} is
equal to
\begin{align*}
\left(\mathrm{d}\Pr[\vx_{a_0}^0 = x_1^t] + \mathrm{d}\Pr[\vx_{a_0}^0 = f_v(x_1^t)]\right)
\frac{w(x_1^t)}{w(x_1^t)+w(f_v(x_1^t))} 
&= \left(\frac{w(x_1^t)}{2\vol(\beta C)} + \frac{w(f_v(x_1^t))}{2\vol(\beta
C)}\right)\frac{w(x_1^t)}{w(x_1^t)+w(f_v(x_1^t))} \\ 
&= \frac{w(x_1^t)}{2\vol(\beta C)} = \mathrm{d}\Pr[\vx_{a_0}^0 = x_1^t]
\end{align*}
Hence the probabilities in Equations~\eqref{eq:run-a-prob} and~\eqref{eq:corr-prob} are equal as needed. Lastly, note
that when conditioning on $\mathcal{G}$, both of the corresponding probabilities are divided by $\Pr[\mathcal{G}]$, and
hence equality is maintained.

\paragraph{Output Analysis:}

\paragraph{Claim 4:} Let $T$ denote the last stage of the sieve. Then conditioned on the event $\mathcal{G}$, with
probability at least $1-\left(\frac{2}{3}\right)^{\frac{1}{2}N_G}$ there exists a lattice vector $w \in
\set{\vy_i^T-\vx_i^T-(\vy_j^T-\vx_j^T): i,j \in [\vN_t]}$ satisfying \linebreak $(\dagger)~ w \in L \setminus M$, $w-v
\in M \cap L$ and $\|w-v\|_C < \eps \beta$. Furthermore, any lattice vector satisfying $(\dagger)$ is a
$(1+\eps)$-approximate solution to SAP.

Let $(x_1^T,y_1^T),\dots,(x_{N_t}^T,y_{N_T}^T)$ denote any valid instantiation of the stage $T$ variables corresponding
to a \emph{good} run of the algorithm (i.e. one belonging to $\mathcal{G}$). Let $(\bar{x}_i^T,y_i^T) =
(c_v(x_i^T,y_i^T),y_i^T)$ for $i \in [N_T]$. By claim 3, it suffices to prove the claim for the pairs
$(F_v(\bar{x}_1^T),y_1^T),$ $\dots$ $,(F_v(\bar{x}_{N_T}^T),y_{N_T}^T)$. This follows since the probability of
``success'' (i.e. the existence of the desired vector) conditioned on $\mathcal{G}$, is simply an average over all
instantiations above of the conditional probability of ``success''.

Since our instantiation corresponds to a \emph{good} run, by Claim $2$ we have at least $N_G$ \emph{good} pairs
in stage $T$. Since $c_v$ preserves \emph{good} pairs, the same holds true for $(\bar{x}_i^T,y_i^T)$, $i \in [N_t]$.
For notational convenience, let us assume that the pairs $(\bar{x}_i^T,y_i^T)$, $i \in [N_G]$ are all \emph{good}.
We note then that $f_v(\bar{x}_i^T) = \bar{x}_i^T \pm v$ and $f_v(f_v(\bar{x}_i^T)) = \bar{x}_i^T$ for $i \in [N_G]$.

First, since $T$ is last the stage, we know that $\|y_i^T\|_C^* \leq 3\beta$ for $i \in [N_G]$. Next, for $i \in [N_G]$, by definition of $c_v$ we have that
\[
\|y_i^T-\bar{x}_i^T\|_C^* = \min \set{\|y_i^T-x_i^T\|_C^*, \|y_i^T-f_v(x_i^T)\|_C^*}
\]
Let $s = s(y_i^T)$, i.e. $\|y_i^T\|_C^* = \|y_i^T\|_{sC}$. Since $(x_i^T,y_i^T)$ is \emph{good} at least one
of $-x_i^T,-f_v(x_i^T) \in \beta sC$. Without loss of generality, we assume $-x_i^T \in \beta sC$. Therefore,
we get that
\begin{equation}
\label{eq:size-bnd}
\|y_i^T-\bar{x}_i^T\|_C^* \leq \|y_i^T-x_i^T\|_C^* \leq \|y_i^T-x_i^T\|_{sC} \leq \|y_i^T\|_{sC} + \|-x_i^T\|_{sC} \leq 3\beta + \beta = 4\beta
\end{equation}
Let $S$ denote the set $\set{y_i^T-\bar{x}_i^T: i \in [N_G]}$. Since $\bar{x}_i^T \equiv y_i^T \bmod B$, we note that $S
\subseteq L$. Also by Equation~\eqref{eq:size-bnd} we have that $S \subseteq 4\beta(C \cup -C) \cap L$. Let $\Lambda
\subseteq S$ denote a maximal subset such that $x + \intr(\frac{\eps}{2} \beta (C \cap -C)) \cap y + \intr(
\frac{\eps}{2} \beta (C \cap -C)) = \emptyset$ for distinct $x,y \in \Lambda$. Since $S \subseteq 4 \beta(C \cup -C)$,
we see that for $x \in S$
\[
x + \frac{\eps}{2} \beta (C \cap -C) \subseteq 4 \beta (C \cup -C) + \frac{\eps}{2} \beta (C \cap -C)
\subseteq \left(4 + \frac{\eps}{2}\right) \beta (C \cup -C)
\]
Therefore we see that
\begin{align*}
|\Lambda| &\leq \frac{\vol((4+\frac{\eps}{2})\beta (C \cup -C))}{\vol(\frac{\eps}{2} ~ \beta (C \cap -C))} 
          = \left(\frac{8+\eps}{\eps}\right)^n \frac{\vol(C \cup -C)}{\vol(C \cap -C)} \\
          &\leq \left(\frac{8 + \eps}{\eps}\right)^n \frac{2 \vol(C)}{\gamma^n \vol(C)} 
           \leq 2\left(\frac{9}{\gamma \eps}\right)^n \leq \frac{1}{2}N_G
\end{align*}
Since $\Lambda$ is maximal, we note that for any $x \in S$, there exists $y \in \Lambda$ such that
\begin{equation}
\label{eq:covers}
\intr(x + \frac{\eps}{2} \beta (C \cap -C)) \cap \intr(y + \frac{\eps}{2} \beta (C \cap -C)) \neq \emptyset
\Leftrightarrow \|x-y\|_{C \cap -C} < \eps \beta
\end{equation}
Let $c_1,\dots,c_{|\Lambda|} \in [N_G]$, denote indices such that $\Lambda = \set{y_{c_i}^T-\bar{x}_{c_i}^T: 1 \leq i \leq
|\Lambda|}$, and let $C = \set{c_j: 1 \leq j \leq |\Lambda|}$. For $j \in \set{1,\dots,|\Lambda|}$, recursively define
the sets
\begin{equation}
\label{eq:partition}
I_j = \set{i \in [N_G]: \|(y_i^T-\bar{x}_i^T)-(y_{c_j}^T-\bar{x}_{c_j}^T)\|_{C \cap -C} < \eps \beta} 
\setminus \left(C ~ \cup ~ (\cup_{k=1}^{j-1} I_k)\right)
\end{equation}
Given Equation \ref{eq:covers}, we have by construction that the sets $C,I_1,\dots,I_{|\Lambda|}$ partition $[N_G]$.
For each $j \in \set{1,\dots,|\Lambda|}$, we examine the differences
\[
S_j = \pm \set{ (y_i^T - F_v(\bar{x}_i^T)) - (y_{c_j}^T-F_v(\bar{x}_{c_j}^T)): i \in I_j}
\]
We will show that $S_j$ fails to contain a vector satisfying $(\dagger)$ with probability at most
$\left(\frac{2}{3}\right)^{|I_j|}$. First we note that $S_j \subseteq L$ since $y_i^T \equiv \bar{x}_i^T \equiv
F_v(\bar{x}_i^T) \bmod B$ for $i \in [N_G]$.

We first condition on the value of $F_v(\bar{x}_{c_j}^T)$ which is either $\bar{x}_{c_j}^T,\bar{x}_{c_j}^T-v$ or
$\bar{x}_{c_j}^T+v$. We examine the case where $F_v(\bar{x}_{c_j}^T) = \bar{x}_{c_j}^T$, the analysis for the other two
cases is similar. Now, for $i \in I_j$, we analyze the difference
\begin{equation}
\label{eq:rand-vec}
y_i^T-F_v(\bar{x}_i^T) - (y_{c_j}^T-\bar{x}_{c_j}^T)
\end{equation}
Let $\delta_i = (y_i^T-\bar{x}_i^T) - (y_{c_j}^T-\bar{x}_{c_j}^T)$. Depending on the output of $F_v(\bar{x}_i^T)$, note
that the vector~\eqref{eq:rand-vec} is either (a) $\delta_i$ or of the form (b) $\pm v + \delta_i$ (since
$f_v(\bar{x}_i^T) = \bar{x}_i^T \pm v)$. We claim that a vector of form (b) satisfies $(\dagger)$. To see this, note
that after possibly negating the vector, it can be brought to the form $v \pm \delta_i \in L$, where we have that 
\begin{equation}
\label{eq:near-opt}
\|\pm \delta_i\|_C < \|\pm \delta_i\|_{C \cap -C} = \|\delta_i\|_{C \cap -C} 
                   < \eps \beta \leq \eps \lambda(C,L,M) < \lambda(C,L,M) \text{,}
\end{equation}
since $i \in I_j$ and $\eps \leq \frac{1}{2}$. Since $\delta_i \in L$ and $\|\delta_i\|_C < \lambda(C,L,M)$, we must
have that $\pm \delta_i \in M \cap L$. Next, since $v \in L \setminus M$ and $\pm \delta_i \in M$, we have that $v \pm
\delta_i \in L \setminus M$. Lastly, note that 
\[
\|v \pm \delta_i\|_C \leq \|v\|_C + \|\pm \delta_i\|_C < \lambda(C,L,M) + \eps \beta \leq (1+\eps) \lambda(C,L,M)
\]
as required.

Now the probability the vector in~\eqref{eq:rand-vec} is of form (b) is
\[
\Pr[F_v(\bar{x}_i^T) = f_v(\bar{x}_i^T)] = \frac{w(f_v(\bar{x}_i^T))}{w(\bar{x}_i^T) + w(f_v(\bar{x}_i^T))} \geq
\frac{1}{3}
\]
since for any $x \in \beta(C \cup -C)$ we have that $1 \leq w(x) \leq 2$. Since each $i \in I_j$ indexes a vector in
$S_j$ not satisfying $(\dagger)$ with probability at most $1-\frac{1}{3} = \frac{2}{3}$, the probability that $S_j$
contains no vector satisfying $(\dagger)$ is at most $\left(\frac{2}{3}\right)^{|I_r|}$ (by independence) as needed.

Let $F_j$, $j \in \set{1,\dots,|\Lambda|}$, denote the event that $S_j$ does not contain a vector satisfying
$(\dagger)$. Note that $F_j$ only depends on the pairs $(F_v(\bar{x}_{c_j}^T),y_{c_j}^T)$ and $(F_v(\bar{x}_i^T),y_i^T)$
for $i \in I_j$. Since the sets $I_1,\dots,I_{|\Lambda|},C$ partition $[N_G]$, these dependencies are all disjoint, and
hence the events are independent. Therefore the probability that none of $S_1,\dots,S_{|\Lambda|}$ contain
a vector satisfying $(\dagger)$ is at most
\[
\Pr[\cap_{j=1}^{|\Lambda|} F_j] \leq \prod_{j=1}^{|\Lambda|} \left(\frac{2}{3}\right)^{|I_j|}
= \left(\frac{2}{3}\right)^{N_G - |\Lambda|} \leq \left(\frac{2}{3}\right)^{\frac{1}{2}N_G}
\]
as needed.

\paragraph{Runtime and Failure Probability:} We first analyze the runtime. First, we make $O(n \log \frac{R}{r})$
guesses for the value of $\Lambda(C,L,M)$. We run the ShortVectors algorithm once for each such guess $\beta$. During
one iteration of the sieving algorithm, we first generate $\vN_0 ~~ \eta$-uniform samples from $\beta C$ (line 3). By
Theorem \ref{thm:unif-sampler}, this takes $\poly(n,\ln \frac{1}{\eta}, \ln \beta, \ln R, \ln r)$ time per sample. We
also mod each sample by the basis $B$ for $L$, which takes $\poly(|B|)$ time ($|B|$ is the bit size of the basis). Next,
by the analysis of Claim $2$, we apply the sieving procedure at most $\ceil{6 \ln \frac{D}{\beta}}$ times (runs of while
loop at line 5), where each iteration of the sieving procedure (line 6) takes at most $O(\vN_0
\left(\frac{5}{\gamma}\right)^n)$ time by Lemma \ref{lem:basic-sieve}. Lastly, we return the set of differences (line
8), which takes at most $O(\vN_0^2)$ time. Now by standard arguments, one has that the values $D$ and $\beta$ (for each
guess) each have size (bit description length) polynomial in the input, i.e. polynomial in $|B|$ (bit size of the basis
of $L$), $n$, $\ln R$, $\ln r$. Since $\vN_0 = O(\ln(\frac{D}{\beta})(\frac{36}{\gamma^2 \eps})^n)$, we have that the
total running time is 
\[
\poly(n,\ln R, \ln r, |B|, \nicefrac{1}{\eps})~O\left(\frac{1}{\gamma^4 \eps^2}\right)^n \text{ as needed. }
\]

We now analyze the success probability. Here we only examine the guess $\beta$, where $\beta \leq \lambda(C,L,M) \leq
\frac{3}{2}\beta$. Assuming perfectly uniform samples over $\beta C$, by the analysis of Claim $4$, we have that
conditioned on $\mathcal{G}$, we fail to output a $(1+\eps)$ approximate solution to SAP with probability at most
$\left(\frac{2}{3}\right)^{\frac{1}{2}N_G}$. Hence, under the uniform sampling assumption, the total probability of failure
is at most 
\[
\left(\frac{2}{3}\right)^{\frac{1}{2}N_G} + \Pr[\mathcal{G}^c] \leq \left(\frac{2}{3}\right)^{\frac{1}{2}N_G}
+ e^{-\frac{1}{48}\gamma^n\vN_0} \ll 2^{-(n+1)}
\]
by Claim $2$. When switching to $\eta$-uniform samples, as argued in the preliminary analysis, this failure probability
increases by at most the total variation distance, i.e. by at most $\eta \vN_0 = 2^{-(n+1)}$. Therefore, the algorithm
succeeds with probability at least $1 - 2^{-(n+1)} - 2^{-(n+1)} = 1 - 2^{-n}$ as needed.

\paragraph{Exact SAP:} Here we are given the guarantee that $\lambda(C,L,M) \leq t \lambda_1(C,L)$, and we wish to use
our SAP solver to get an exact minimizer to the SAP. To solve this, we run the approximate SAP solver on $C,L,M$ with
parameter $\eps = \frac{1}{t}$, which takes $O(t^2/\gamma^4)^n$ time. Let $v \in L \setminus M$ be a lattice vector
satisfying $\|v\|_C = \lambda(C,L,M)$. By Claim $4$, with probability at least $1-2^{-n}$ we are guaranteed to output a
lattice vector $w \in L \setminus M$, such that 
\[
\|w-v\|_C < \eps \lambda(C,L,M) \leq \left(\frac{1}{t}\right) t \lambda_1(C,L) = \lambda_1(C,L)
\]
However, since $w-v \in L$ and $\|w-v\|_C < \lambda_1(C,L)$, we must have that $w-v = 0$. Therefore $w = v$
and our SAP solver returns an exact minimizer as needed.
\end{proof}

\subsection{Closest Vector Problem}
\label{sec:cvp}

In this section, we present a reduction from Approximate-CVP to Approximate-SAP for general semi-norms. In
\cite{DBLP:conf/icalp/BlomerN07}, it is shown that $\ell_p$ CVP reduces to $\ell_p$ SAP in one higher dimension. By
relaxing the condition that the lifted SAP problem remain in $\ell_p$, we give a very simple reduction which reduces CVP
in any semi-norm to SAP in one higher dimension under a different semi-norm that is essentially as symmetric. Given
the generality of our SAP solver, such a reduction is suffices.

\begin{theorem}[Approximate-CVP] Take $x \in \R^n$. Then for any $\eps \in (0,\frac{1}{3})$, $y \in L$ satisfying
$\|y-x\|_C \leq (1+\eps)d_C(L,x)$ can be computed in time $O(\frac{1}{\gamma^4 \eps^2})^n$ with probability at least
$1-2^{-n}$. Furthermore, if $d_C(L,x) \leq t \lambda_1(C,L)$, $t \geq 2$, then a vector $y \in L$ satisfying
$\|y-x\|_C = d_C(L,x)$ can be computed in time $O(\frac{t^2}{\gamma^4})^n$ with probability at least $1-2^{-n}$.
\label{thm:approx-cvp}
\end{theorem}
\begin{proof}
To show the theorem, we use a slightly modified version of Kannan's lifting technique to reduce CVP to SAP.
Let us define $L' \subseteq \R^{n+1}$ as the lattice generated by $(L,0)$ and $(-x,1)$.

In the standard way, we first guess a value $\beta > 0$ satisfying $\beta \leq d_C(L,x) \leq \frac{3}{2} ~ \beta$. Now
let $C' = C \times [-\frac{1}{\beta},\frac{1}{2\beta}]$. For $(y,z)$, $y \in \R^n$, $z \in \R$, we have that
\[
\|(y,z)\|_{C'} = \max \set{\|y\|_C, \beta z, -2 \beta z}
\]
Also, note that $C' \cap -C' = (C \cap -C) \times [-\frac{1}{2\beta}, \frac{1}{2\beta}]$. Now
we see that $\vol(C' \cap -C') = \frac{1}{\beta} \vol(C \cap -C)$ and $\vol(C') = \frac{3}{(2\beta)} \vol(C)$.
Therefore, we get that 
\[
\vol(C' \cap -C') = \frac{1}{\beta} \vol(C \cap -C) \geq \frac{1}{\beta} ~ \gamma^n \vol(C) = 
\frac{2}{3} ~ \gamma^n \vol(C') 
\]
Hence $C'$ is $\gamma(1-1/n)$-symmetric.  Let $M = \set{y \in \R^{n+1}: y_{n+1} = 0}$. Define $m: L \rightarrow L'
\setminus M$ by $m(y) = (y-x,1)$, where it is easy to see that $m$ is well-defined and injective. Define 
\[
S = \set{y \in L: \|y-x\|_C \leq (1+\eps)d_C(L,x)} \quad \text{ and } \quad  S' = \set{y \in L' \setminus M: \|y\|_C'
\leq (1+\eps) \lambda(C',L',M)} \text{.}
\]
We claim that $m$ defines a norm preserving bijection between $S$ and $S'$. Taking $y \in L$, we see that
\[
\|m(y)\|_{C'} = \|(y-x,1)\|_{C'} = \max \set{ \|y-x\|_C, \beta, -2\beta} = \|y-x\|_C
\]
since $\beta \leq d_C(L,x) \leq \|y-x\|_C$ by construction. So we have that $\|m(y)\|_{C'} = \|y-x\|_C$, and hence
$\lambda(C',L',M) \leq \inf_{y \in L} \|y-x\|_C = d_C(L,x)$. Next take $(y,z) \in L' \setminus M$, $y \in \R^n$, $z \in
\R$, such that $\|(y,z)\|_{C'} \leq (1+\eps) \lambda(C',L',M)$. We claim that $z = 1$. Assume not, then since $(y,z) \in
L' \setminus M$, we must have that either $z \geq 2$ or $z \leq -1$. In either case, we have that
\[
\|(y,z)\|_{C'} = \max \set{\|y\|_C, \beta z, - 2\beta z} \geq \max \set{ \beta z, -2\beta z} \geq 2\beta 
\]
Now since $\beta \leq d_C(L,x) \leq \frac{3}{2} ~\beta$, $\eps \in (0,\frac{1}{3})$, and that
$\lambda(C',L',M) \leq d_C(L,x)$, we get that
\[
\|(y,z)\|_{C'} \geq 2 \beta = (1+\frac{1}{3})(\frac{3}{2}\beta) \geq (1+\frac{1}{3})d_C(L,x) > (1+\eps)d_C(L,x) \geq
(1+\eps)\lambda(C',L',M)
\]
a clear contradiction to our initial assumption. Since $z = 1$, we may write $y = w-x$ where $w \in L$. Therefore,
we see that
\[
\|(y,z)\|_{C'} = \|(w-x,1)\|_{C'} = \max \set{\|w-x\|_C, \beta, -2\beta} = \|w-x\|_C
\]
since $\|w-x\|_C \geq d_C(L,x) \geq \beta$. So we have that $(1+\eps) \lambda(C',L',M) \geq \|(y,z)\|_{C'} = \|w-x\|_C \geq
d_C(L,x)$. Since the previous statement still holds when choosing $\eps = 0$, we must have that $\lambda(C',L',M) \geq
d_C(L,x)$ and hence $\lambda(C',L',M) = d_C(L,x)$.

From the above, for $y \in S$, we have that $\|m(y)\|_{C'} = \|y-x\|_C \leq (1+\eps) d_C(L,x) = (1+\eps)
\lambda(C',L',M)$, and hence $m(y) \in S$ as needed. Next if $(y,z) \in S'$, from the above we have that $z = 1$, and
hence $(y,z) = (w-x,1)$ where $w \in L$. Therefore $(y,z) = m(w)$, where $\|w-x\|_C = \|(y,z)\|_{C'} \leq
(1+\eps)\lambda(C',L',M) = (1+\eps)d_C(L,x)$, and hence $w \in S$. Since the map $m$ is injective, we get that
$m$ defines a norm preserving bijection between $S$ and $S'$ as claimed.

Hence solving $(1+\eps)$-CVP with respect to $C,L,x$ is equivalent to solving $(1+\eps)$-SAP with respect to $C',L',M$.
We get the desired result by applying $1+\eps$ approximation algorithm for SAP described in theorem \ref{thm:approx-sap}.

For exact CVP, we are given the guarantee that $d_C(L,x) \leq t \lambda_1(C,L)$. From analysis above, we see that
\[
\lambda_1(C',L') = \min \set{\lambda_1(C',L',M), \inf_{y \in L \setminus \set{0}} \|(y,0)\|_{C'}}
                 = \min \set{d_C(L,x), \lambda_1(C,L)}
\]
Therefore 
\[
\lambda(C',L',M) = d_C(L,x) = \min \set{d_C(L,x), t \lambda_1(C,L)} \leq t \min \set{d_C(L,x), \lambda_1(C,L)}
               = t \lambda_1(C',L')
\]
Hence we may again use the SAP solver in theorem \ref{thm:approx-sap} to solve the exact CVP problem in
$O(\frac{t^2}{\gamma^4})^n$ time with probability at least $1-2^{-n}$ as required.
\end{proof}


%% file: conclusions.tex
\section{Conclusions and Open Problems}
\label{sec:conclusion}

In this paper, we have shown that an approximate version of Integer Programming (IP) can be solved via an extension of
the AKS sieving techniques to general semi-norms. Furthermore, we give algorithms to solve both the
$(1+\eps)$-Approximate Subspace Avoiding Problem (SAP) and Closest Vector Problem (CVP) in general semi-norms using
these techniques. Due to the reliance on a probabilistic sieve, the algorithms presented here only guarantee the
correctness of their outputs with high probability. In \cite{Dadush-et-al-SVP-11}, it was shown that this shortcoming
can sometimes be avoided, by giving a Las Vegas algorithm for SVP in general norms (which is deterministic for $\ell_p$
norms) achieving similar asymptotic running times as the AKS sieve based methods. The following question is still open:

\paragraph{{\Large Problem:}} Does there exists a Las Vegas or deterministic algorithm for $(1+\eps)$-SAP or CVP
in general (semi-)norms achieving the same asymptotic running time as the AKS sieve based methods?
\vspace{1em}

As for potential improvements in the complexity of $(1+\eps)$-SAP / CVP, the following question is open:

\paragraph{{\Large Problem:}} Can the complexity of the AKS sieve based methods for $(1+\eps)$-SAP / CVP
in general (semi-)norms be reduced to $O(\frac{1}{\gamma^2\eps})^n$?

%% file: ack.tex
\section{Acknowledgments}
I would like to thank my advisor Santosh Vempala for useful discussions relating to this problem.


%% file: appendix.tex
\section{Appendix}
\label{sec:appendix}

\begin{proof}[Proof of Lemma \ref{lem:barycenter-approx} (Approx. Barycenter)]
Let $X_1,\dots,X_N$ denote iid uniform samples over $K \subseteq \R^n$, where $N = \left(\frac{2c n}{\eps}\right)^2$.
We will show that for $b = \frac{1}{N} \sum_{i=1}^n X_i$, that the following holds
\begin{equation}
\label{eq:far-away-prob}
\Pr[\|\pm(b-b(K))\|_{K-b(K)} > \eps] \leq 4^{-n}
\end{equation}
Since the above statement is invariant under affine transformations, we may assume $K$ is isotropic,
i.e. $b(K) = \E[X_1] = 0$, the origin, and $\E[X_1X_1^t] = I_n$, the $n \times n$ identity. Since $K$ is
isotropic, we have that $B_2^n \subseteq K$ (see \cite{KLS95}). Therefore to show~\eqref{eq:far-away-prob}
it suffices to prove that $\Pr[\|b\|_2 > \eps] \leq 4^{-n}$. Since the $X_i$'s are iid isotropic random
vectors, we see that $\E[b] = \frac{1}{N} \sum_{i=1}^N \E[X_i] = 0$ and
\[
\E[bb^t] = \frac{1}{N^2} \sum_{i,j \in [N]} E[X_iX_j^t] = \frac{1}{N^2} \sum_{i=1}^n E[X_iX_i^t] = \frac{1}{N} I_n
\]
Now since the $X_i$s are log-concave, we have that $b$ is also log-concave (since its distribution is a convolution of
log-concave distributions). Now, given that $b$ has covariance matrix $\frac{1}{N} I_n$, by the concentration inequality of Paouris \cite{Paouris-Concentration-06}, we have that
\[
\Pr[\|b_2\|_2 > \eps] = \Pr[\|b_2\|_2 > 2c\frac{n}{\sqrt{N}}] < e^{-2n} < 4^{-n}
\]
as claimed. To prove the theorem, we note that when switching the $X_i$'s from truly uniform to $4^{-n}$ uniform,
the above probability changes by at most $\frac{c n^2}{\eps^2} 4^{-n}$ by Lemma \ref{lem:tvd}. Therefore
the total error probability under $4^{-n}$-uniform samples is at most $2^{-n}$ as needed.
\end{proof}

\begin{proof}[Proof of Lemma \ref{lem:cent-comp} (Estimates for semi-norm recentering)]
We have $z \in \R^n$, $x,y \in K$ satisfying \linebreak $(\dagger) ~ \|\pm(x-y)\|_{K-y} \leq \alpha < 1$. We
prove the statements as follows:
\begin{enumerate}
\item 
\[
\|z-y\|_{K-y} \leq \tau \Leftrightarrow (z-y) \in \tau(K-y) \Leftrightarrow z \in \tau K + (1-\tau) y
\text{ as needed .}
\]
\item Let $\tau = \|z-x\|_{K-x}$. Then by $(1)$, we have that $z \in \tau K + (1-\tau)x$. Now note
that
\[
(1-\tau)(x-y) \subseteq |1-\tau| \alpha(K-y)
\]
by assumption $(\dagger)$ and $(1)$. Therefore
\begin{align*}
z \in \tau K + (1-\tau)x &= \tau K + (1-\tau)y + (1-\tau)(x-y) 
                   \subseteq \tau K + (1-\tau)y + \alpha |1-\tau|(K-y) \\
                    &= (\tau + \alpha |1-\tau|) K + (1-\tau-\alpha|1-\tau|)y
\end{align*}
Hence by $(1)$, we have that 
\[
\|z-y\|_{K-y} \leq \tau + \alpha |1-\tau| = \|z-x\|_{K-x} + \alpha |1-\|z-x\|_{K-x}|
\]
as needed.
\item We first show that
\[
\pm(y-x) \in \frac{\alpha}{1-\alpha}(K-x)
\]
By $(1)$ and $(\dagger)$ we have that
\begin{align*}
(x-y) \in \alpha(K-y) &\Leftrightarrow (x-y) - \alpha(x-y) \in \alpha(K-y) - \alpha(x-y) \\
                      &\Leftrightarrow (1-\alpha)(x-y) \in \alpha(K-x)   
                      \Leftrightarrow (x-y) \in \frac{\alpha}{1-\alpha}(K-x)
\end{align*}
as needed. Next since $0 \leq \alpha \leq 1$, we have that $|1-2\alpha| \leq 1$. Therefore by $(\dagger)$ we have that
\[
(1-2\alpha)(y-x) \in |1-2\alpha| \alpha (K-y) \subseteq \alpha (K-y)
\]
since $0 \in K-y$. Now note that
\begin{align*}
(1-2\alpha)(y-x) \in \alpha(K-y) &\Leftrightarrow (1-2\alpha)(y-x) + \alpha(y-x) \in \alpha(K-y) + \alpha(y-x) \\
				 &\Leftrightarrow (1-\alpha)(y-x) \in \alpha(K-x)
				 \Leftrightarrow (y-x) \in \frac{\alpha}{1-\alpha}(K-x)
\end{align*}
as needed.

Let $\tau = \|z-y\|_{K-y}$. Then by $(1)$, we have that $z \in \tau K + (1-\tau)y$. Now note that
\begin{align*}
z \in \tau K + (1-\tau)y &= \tau K + (1-\tau)x + (1-\tau)(y-x) 
                   \subseteq \tau K + (1-\tau)x + \frac{\alpha}{1-\alpha}|1-\tau| (K-x) \\
                   &= (\tau + \frac{\alpha}{1-\alpha} |1-\tau|) K + (1-\tau-\frac{\alpha}{1-\alpha}|1-\tau|)x
\end{align*}
Hence by $(1)$, we have that 
\[
\|z-x\|_{K-x} \leq \tau + \frac{\alpha}{1-\alpha} |1-\tau| = \|z-y\|_{K-y} + \frac{\alpha}{1-\alpha} |1-\|z-y\|_{K-y}|
\]
as needed.
\end{enumerate}
\end{proof}

\begin{proof}[Proof of Corollary \ref{cor:sym-stable} (Stability of symmetry)]
We claim that $(1-\|x\|_K)(K \cap -K) \subseteq K-x \cap x-K$. Take $z \in K \cap -K$, then note that
\begin{align*}
\|x + (1-\|x\|_K)z\|_K &\leq \|x\|_K + (1-\|x\|_K)\|z\|_K \leq \|x\|_K + (1-\|x\|_K)\|z\|_{K \cap -K} \\
                       &\leq \|x\|_K + (1-\|x\|_K) = 1
\end{align*}
hence $x + (1-\|x\|_K)(K \cap -K) \subseteq K \Leftrightarrow (1-\|x\|_K)(K \cap -K) \subseteq K-x$.
Next note that
\begin{align*}
\|-x + (1-\|x\|_K)z\|_{-K} &\leq \|-x\|_{-K} + (1-\|x\|_K)\|z\|_{-K} \leq \|x\|_K + (1-\|x\|_K)\|z\|_{K \cap -K} \\
                           &\leq \|x\|_K + (1-\|x\|_K) = 1
\end{align*}
hence $-x + (1-\|x\|_K)(K \cap -K) \subseteq -K \Leftrightarrow (1-\|x\|_K)(K \cap -K) \subseteq x-K$, as needed.
Now we see that
\[
\vol((K-x) \cap (x-K)) \geq \vol((1-\|x\|_K)(K \cap -K)) = (1-\|x\|_K)^n \vol(K \cap -K)
\]
and so the claim follows from Theorem \ref{thm:symmetrize}.
\end{proof} 

\begin{proof}[Proof of Lemma \ref{lem:li} (Intersection Lemma)]
Since $\|v\|_C \leq \frac{3}{2} \beta$, we see that $\frac{v}{2} \in \nicefrac{3}{4} C$. Now we get that
\[
\frac{v}{2} + \nicefrac{1}{4} \beta (C \cap -C) \subseteq \nicefrac{3}{4} \beta C + \nicefrac{1}{4} \beta C
                                                   = \beta C
\]
Furthermore since $\|\frac{v}{2} - v\|_{-C} = \|-\frac{v}{2}\|_{-C} = \nicefrac{1}{2} \|v\|_C \leq
\nicefrac{3}{4} \beta$, we also have that $\frac{v}{2} \in v - \nicefrac{3}{4} C$. Therefore
\[
\frac{v}{2} + \frac{1}{4} \beta (C \cap -C) \subseteq (v - \frac{3}{4} \beta C) + \frac{1}{4} \beta (-C)
                                                   = v - \beta C
\]
We therefore conclude that
\[
\frac{\vol(\beta C \cap (v - \beta C))}{\vol(\beta C)} \geq \frac{\vol(\frac{v}{2} + \nicefrac{1}{4}\beta (C \cap
-C))}{\vol(C)} = \left(\frac{1}{4}\right)^n \frac{\vol(C \cap -C)}{\vol(C)} \geq \left(\frac{\gamma}{4}\right)^n
\]
as needed.

For the furthermore, we remember that $\|x\|_C = \inf \set{s \geq 0: x \in sC} = \sup \set{\pr{x}{y}: y \in C^*}$ and that
$(-C)^* = -C^*$. Now assume there exists $x \in C^+_v \cap C^-_v$. Then $x = v - \beta k_1 = \beta k_2 -v$ where
$k_1,k_2 \in C$. Choose $y \in C^*$ such that $\pr{y}{v} = \|v\|_C$. Note that
\begin{align*}
\pr{y}{v-\beta k_1 - (\beta k_2 - v)} &= 2\pr{y}{v} - \beta (\pr{y}{k_1} + \pr{y}{k_2}) = 
                                           2\|v\|_C - \beta (\pr{y}{k_1}+\pr{y}{k_2}) \\
                                 &\geq 2\|v\|_C - \beta (\|k_1\|_C + \|k_2\|_C) \geq 2\beta - 2\beta = 0
\end{align*}
Since $v-\beta k_1 - (\beta k_2 - v) = x - x = 0$ by construction, all of the above inequalities must hold at equality.
In particular, we must have that $1 = \|k_1\|_C = \|k_2\|_C = \pr{y}{k_1} = \pr{y}{k_2}$. Since $-y \in (-C)^*$,
we know that
\[
v - \beta C \subseteq \set{x \in \R^n: \pr{-y}{x-v} \leq \beta}
\]
and since $\pr{-y}{(v-\beta k_1) - v} = \beta \pr{y}{k_1} = \beta$, we must have that $v - \beta k_1 \in \partial
C_v^+$. Via a symmetric argument, we get that $\beta k_2 - v \in \partial C_v^-$. Therefore $C_v^+ \cap C_v^- \subseteq
\partial C_v^+ \cap \partial C_v^- \Leftrightarrow \intr(C_v^+) \cap \intr(C_v^-) = \emptyset$, as needed.
\end{proof}